\documentclass{elsarticle}

\usepackage{amssymb,amsthm,latexsym,amsmath,graphics,setspace}
\usepackage{epsfig,epsf,amsfonts}
\usepackage{color,multicol,colordvi,graphicx,multirow}

\newtheorem{theorem}{Theorem}[section]
\newtheorem{lemma}[theorem]{Lemma}

\newtheorem{proposition}[theorem]{Proposition}

\newcommand{\paren}[1]{\left(#1\right)}
\newcommand{\jump}[1]{\left[#1\right]}
\newcommand{\D}[2]{\frac{d#1}{d#2}}
\newcommand{\PD}[2]{\frac{\partial#1}{\partial#2}}
\newcommand{\PDD}[3]{\frac{\partial^{#1}{#2}}{\partial{#3}^{#1}}}

\newcommand{\at}[2]{\left. #1 \right|_{#2}}

\newcommand{\mb}[1]{\mathbf{#1}}
\newcommand{\bm}[1]{\boldsymbol{#1}}
\newcommand{\abs}[1]{\left\lvert #1 \right\rvert}

\newcommand{\dual}[2]{\left\langle #1,#2 \right\rangle}

\numberwithin{equation}{section}

\begin{document}

\begin{frontmatter}

\title{A Model of Electrodiffusion 
and Osmotic Water Flow and its Energetic Structure}

\author[umn]{Yoichiro Mori\corref{cor}}
\ead{ymori@math.umn.edu}
\author[psu]{Chun Liu}
\ead{liu@math.psu.edu}
\author[rush]{Robert S. Eisenberg}
\ead{beisenbe@rush.edu}

\cortext[cor]{Corresponding Author}
\address[umn]{School of Mathematics, University of Minnesota, 
Minneapolis, MN 55414, U.S.A.}
\address[psu]{Department of Mathematics, Pennsylvania State University,
University Park, PA 16802, U.S.A.}
\address[rush]{Department of Molecular Biophysics and Physiology, 
Rush University Medical Center, Chicago, IL 60612, U.S.A.}

\begin{abstract}
We introduce a model for ionic electrodiffusion and osmotic water flow 
through cells and tissues. 
The model consists of a system of partial differential equations 
for ionic concentration and fluid flow with interface conditions 
at deforming membrane boundaries. 
The model satisfies a natural energy equality, 
in which the sum of the entropic, elastic and electrostatic 
free energies are dissipated through viscous, electrodiffusive 
and osmotic flows. 
We discuss limiting models when certain dimensionless parameters are small. 
Finally, we develop a numerical scheme for the one-dimensional 
case and present some simple applications of our model to cell volume control. 
\end{abstract}

\end{frontmatter}

\section{Introduction}
Systems in with important electrodiffusion and osmotic water flow 
are found throughout life 
\cite{pappenheimer1987silver,boron2008medical,davson1970textbook}. 
Such systems include brain ionic homeostasis 
\cite{somjen2004ions,kahle2009molecular}, 
fluid secretion by epithelial systems \cite{hill2008fluid}, 
electrolyte regulation in the kidney 
\cite{koeppen2007renal,layton2009mammalian}, 
fluid circulation in ocular systems 
\cite{mathias2007lens,fischbarg2005mathematical}, 
and water uptake by plants \cite{taiz2010plant}. 

Mathematical models of electrodiffusion and/or osmosis have 
been proposed and used in many physiological contexts, 
and have formed a central topic in biology for a very 
long time \cite{pappenheimer1987silver,Hille,aidley_physiology_1998}. 
Some are simple models using ordinary differential equations 
while others are more detailed in that they include partial 
differential equations (PDEs) describing the spatial variation
of the concentration and flow fields
\cite{KS,hoppensteadt2002modeling,
weinstein1994mathematical,yi2003mathematical,shapiro2001osmotic,
lee2008immersed,mathias1985steady}. 
In this paper, we propose a system of PDEs 
that describes ionic electrodiffusion and osmotic  
water flow at the cellular level. 
To the best of the authors' knowledge, this is the first model 
in which osmotic water flow and electrodiffusion have been treated 
within a unified framework including cells
with deformable 
and capacitance-carrying membranes.
A salient feature of our model is that it possesses a natural thermodynamic 
structure; it satisfies a free energy equality.
As such, the present work may be viewed 
as a generalization of the classical treatment 
of osmosis and electrodiffusion in irreversible thermodynamics 
to spatially extended systems 
\cite{katzir1965nonequilibrium,kedem1958thermodynamic,kjelstrup2008non}.
 
To introduce our approach, we first focus attention on uncharged systems. 
In Section \ref{diffosm}, we treat the case in which the
diffusing chemical species carry no electric charge.
We write down equations that are satisfied by the water velocity
field $\mb{u}$, the chemical concentrations $c_k, k=1,\cdots,N$ 
and the membrane 
position $\mb{X}$. The model is shown to satisfy a free energy 
equality in which the sum of the entropic free energy and 
the elastic energy of the membrane is dissipated through
viscous water flow, bulk diffusion,
transmembrane chemical fluxes and osmotic water flow. 
One interesting consequence of this analysis 
is that the classical van t'Hoff
law of osmotic pressure arises naturally 
from the requirement that osmotic water flow be dissipative.
We note that models with the similar purpose of 
describing diffusing non-electrolytes and their interaction 
with osmotic water flow across moving membranes, 
have been proposed in the literature
\cite{lee2008immersed,atzberger2009microfluidic,layton2006modeling}
(in the Appendix \ref{leeatz}, we discuss the relationship of our model 
with that of \cite{lee2008immersed, atzberger2009microfluidic}).

In Section \ref{elecdiff}, we extend the model of Section \ref{diffosm}
to treat the case of ionic electrodiffusion. We introduce 
the electrostatic potential $\phi$ which satisfies the Poisson 
equation. The membrane now carries capacitance, which can result 
in a jump in the electrostatic potential across the membrane. 
We shall see that this model also satisfies a free energy equality.
The free energy now includes an electrostatic contribution.
The verification of the free energy equality in this case is not as 
straightforward as in the non-electrolyte case, and  
requires a careful examination of surface terms.

In Section \ref{simple}, we discuss simplifications of our 
model. We make the system dimensionless 
and assess the relative magnitudes of the terms in 
the equations. An important simplification is obtained when 
we take the electroneutral limit. In this case, the electrostatic 
potential becomes a Lagrange multiplier that helps to 
enforce the electroneutrality condition. 

In Section \ref{animal}, we develop a computational scheme 
to simulate the limiting system obtained in the electroneutral limit, 
when the geometry of the cell is assumed spherical. 
As an application, we treat animal cell volume control.

\section{Diffusion of Non-electrolytes and Osmotic Water Flow}
\label{diffosm}
\subsection{Model Formulation}
Consider a bounded domain $\Omega\subset \mathbb{R}^3$ 
and a smooth closed surface $\Gamma\subset \Omega$.
This closed surface divides $\Omega$ into two domains.
Let $\Omega_i\subset \Omega$ be the region bounded by 
$\Gamma$, and let $\Omega_e=\Omega\backslash(\Omega_i\cup \Gamma)$. 
In the context of cell biology, $\Omega_i$ may be identified 
with the intracellular space and $\Omega_e$ the extracellular space.
%(Fig. \ref{tamago}).
Although cell physiological systems of biological cells 
serve as our primary motivation
for formulating the models of this paper, 
this identification is not necessary.

%% \begin{figure}
%% \begin{center}
%% \includegraphics[width=0.3\textwidth]{tamago.eps}
%% \end{center}
%% \caption{$\Omega_i$ and $\Omega_e$ are the extracellular and intracellular 
%% spaces respectively. $\Gamma$ is the membrane, which may move with time. }
%% \label{tamago}
%% \end{figure}

In this section we formulate a system of PDEs
that governs the diffusion of {\em non}-electrolytes and osmotic 
flow of water in the presence of membranes.
In Section \ref{elecdiff}, we shall build upon this model 
to treat the electrolyte case.  

We consider $N$ non-electrolyte chemical species
whose concentrations we call $c_k, k=1,\cdots, N$. 
Let $\omega$ be the entropic part of the free energy per unit 
volume of this solution. The following 
expression for $\omega$ is the most standard choice:
\begin{equation}
\omega_0=\sum_{k=1}^N k_BTc_k\ln c_k.\label{ent}
\end{equation}
This expression is valid when the ionic solution is sufficiently 
dilute and lead to linear diffusion of solute. 
Our calculations, however, do not depend on this choice of $\omega$. 
If the solution in question deviates significantly from 
ideality, other expressions for $\omega$ may be used in place of $\omega_0$.
% Another choice for $\omega$ is discussed in the Appendix \ref{waterpotential}.

Given $\omega$, the chemical potential $\mu_k$ 
of the $k-$th chemical species is given as:
\begin{equation}
\mu_k=\sigma_k, \; \sigma_k\equiv \PD{\omega}{c_k}\label{mukc}.
\end{equation}
We have introduced two symbols $\mu_k$ and $\sigma_k$ in anticipation 
of the discussion of the electrolyte case, where $\mu_k$ and $\sigma_k$
are different.
For water, it is convenient to consider the water potential $\psi_w$, 
the free energy per unit volume, 
rather than the $\mu_w$, the chemical potential (free 
energy per molecule). The water potential and water chemical potential 
are thus related by the relation $v_w\psi_w=\mu_w$ where $v_w$ is the 
volume of water per molecule. We define:
\begin{equation}
\psi_w\equiv \pi_w+p, \; \pi_w=\paren{\omega-\sum_{k=1}^N c_k\sigma_k}
=\paren{\omega-\sum_{k=1}^N c_k\PD{\omega}{c_k}}
\label{muw}
\end{equation}
where $p$ is the pressure.
As we shall see, $p$ will be determined in our model 
via the equations of fluid flow (Eq. \eqref{stokes}). 
The entropic part of $\psi_w$ (or the {\em osmotic pressure}), 
$\pi_w$, is given 
as the {\em negative} multiple
of the (semi)Legendre transform of $\omega$ 
with respect to all the ionic concentrations $c_k$.
This expression for osmotic pressure can be 
found, for example, in \cite{doi1996introduction}.
As we shall see in the proof of Theorem \ref{mainc}, 
the above definition of $\pi_w$ is forced upon us if we insist that our 
model satisfy a free energy dissipation principle. 
 
We begin by writing down the equations of ionic concentration dynamics.
At any point in $\Omega_i$ or $\Omega_e$
\begin{equation}
\PD{c_k}{t}+\nabla\cdot(\mb{u} c_k)=\nabla \cdot \paren{c_k\frac{D_k}{k_BT}\nabla \mu_k}\label{ckeq}
\end{equation}
where $D_k$ is the diffusion coefficient and $\mb{u}$ is the fluid 
velocity field. Ions thus diffuse down the chemical potential 
gradient and are advected with the local fluid velocity.
We have assumed here that cross-diffusion (concentration gradient of 
one species driving the diffusion of another species) is negligible. 

We must supplement these equations with boundary conditions.
Most formulations of non-equilibrium thermodynamic processes seem to 
be confined either to the bulk or to the interface between 
two bulk phases \cite{katzir1965nonequilibrium,degroot1962non,kjelstrup2008non}. 
Here we must couple the equations in the bulk and with 
boundary conditions at the interface, which as a whole 
give us a consistent thermodynamic treatment of diffusion and osmosis.

On the outer boundary $\Gamma_\text{out}=\partial \Omega$, for simplicity, 
we impose no-flux boundary conditions. 
Let us now consider the interfacial boundary conditions on the 
membrane $\Gamma$. 
Since we want to account for osmotic water flow, the membrane $\Gamma$
will deform in time.
Sometimes, we shall use the notation $\Gamma_t$ to make this time dependence
explicit.
Let $\Gamma_\text{ref}$ be the resting or reference configuration of 
$\Gamma$. The membrane will then be a smooth deformation of 
this reference surface.
We may take some (local) coordinate system $\bm{\theta}$ 
on $\Gamma_\text{ref}$, which would serve as a material coordinate 
for $\Gamma_t$.
The trajectory of a point that corresponds to 
$\bm{\theta}=\bm{\theta}_0$
is given by $\mb{X}(\bm{\theta}_0,t)\in \mathbb{R}^3$. For fixed $t$,
$\mb{X}(\cdot,t)$ gives us the shape of the membrane $\Gamma_t$.

Consider a point $\mb{x}=\mb{X}(\bm{\theta},t)$ on the membrane.
Let $\mb{n}$ be the outward unit normal on $\Gamma$ at this point.
The boundary conditions satisfied on the intracellular and extracellular 
faces of the membrane are given by:
\begin{equation}
c_k\paren{\mb{u}-\frac{D_k}{k_BT}\nabla \mu_k}\cdot \mb{n}=c_k\PD{\mb{X}}{t}\cdot \mb{n}+j_k+a_k
\text{ on } \Gamma_i \text{ or } \Gamma_e.\label{ckbc}
\end{equation}
The expression ``on $\Gamma_{i,e}$'' indicates that 
the quantities are to be evaluated on the intracellular 
and extracellular faces of $\Gamma$ respectively.
The term $j_k$ is the passive chemical flux that passes through the membrane
and $a_k$ is the active flux. Fluxes going from 
$\Omega_i$ to $\Omega_e$ is taken to be positive. 
Equation \eqref{ckbc} is just a statement of 
conservation of ions at the moving membrane.
It is easy to check that \eqref{ckeq}
together with \eqref{ckbc} implies conservation of each species.  

The flux $j_k$ is in general a function of
concentrations of all chemical species 
on both sides of the membrane. The chemicals 
are usually carried by channels and transporters, 
and the functional form of $j_k$ describe the 
kinetic features of these carriers.  
As we shall see in Section \ref{cross}, $j_k$ can also be a function 
of the difference in water chemical potential across the membrane. 
 
The passive nature of the $j_k$ is expressed by the following 
inequality:
\begin{equation}\label{jkcond}
[\mu_k]j_k \geq 0, \;
[\mu_k]=\at{\mu_k}{\Gamma_i}-\at{\mu_k}{\Gamma_e}
\end{equation}
where $\at{\cdot}{\Gamma_{i,e}}$ expresses evaluation of quantities 
at the intracellular and extracellular faces of the membrane $\Gamma$
respectively. We shall see that this condition is consistent with 
the free energy identity \eqref{FE}.
For any quantity $w$,  
$[w]=\at{w}{\Gamma_i}-\at{w}{\Gamma_e}$ 
will henceforth always denote the difference between $w$ 
evaluated on the intracellular 
and extracellular faces of $\Gamma$.
A simple example of $j_k$ occurs when $j_k$ is a function only of 
$[\mu_k]$ and satisfies the following conditions:
\begin{equation}\label{monotone}
j_k=j_k([\mu_k]), \; j_k(0)=0, \; \PD{j_k}{[\mu_k]}\geq 0.
\end{equation}
It is easily checked that $j_k$ in this case satisfies \eqref{jkcond}.
We shall see concrete examples in Section \ref{elecdiff}.
Condition \eqref{jkcond} is somewhat restrictive in the 
sense that it is possible to have a ``passive'' flux that does not 
satisfy \eqref{jkcond} when multiple species flow and interact. 
We shall discuss this briefly in Section \ref{cross}.
Condition \eqref{jkcond} needs to be relaxed to describe systems 
in which different chemical species flow through one channel 
or (passive) transporter. Those systems usually couple fluxes 
of different chemical species. They often couple (unidirectional) 
influx and efflux of the same species (symporters and antiporters)
\cite{hille1982transport,Hille,tosteson1989membrane,boron2008medical,davson1970textbook}.  
The active flux $a_k$ is typically due to ionic pump currents
often driven by ATP 
\cite{tosteson1989membrane,boron2008medical,davson1970textbook}.

We now discuss force balance. 
We shall treat the cytosol as a viscous fluid 
and the cell membrane as an elastic surface.
The cell membrane itself is just a lipid bilayer, and 
cannot support a large mechanical load.
The cell membrane is often mechanically reinforced by 
an underlying actin cortex
and overlying system of connective tissue, and in the case of plant cells, 
by an overlying cell wall. If we view these structures 
as part of the membrane, our treatment 
of the membrane being elastic may be a useful simplification. 
We could employ a more complete model of cell mechanics 
incorporating, in particular, the mechanical properties
of the cytoskeleton and extracellular lamina. 
However, our emphasis here
is on demonstrating how osmosis can be seamlessly combined with 
mechanics, and we intentionally keep the mechanical model 
simple to clarify the underlying ideas. 

Consider the equations of fluid flow. The flow field 
$\mb{u}$ satisfies the Stokes equation at any point 
in $\Omega_i$ or $\Omega_e$:
\begin{equation}
\nu\Delta \mb{u}-\nabla p
=0, \; \nabla \cdot \mb{u}=0 \label{stokes}
\end{equation}
where $\nu$ is the viscosity of the electrolyte solution. Note that 
the above equations can also be written as follows:
\begin{equation}
\nabla \cdot \Sigma_m(\mb{u},p)=0,\; \nabla \cdot \mb{u}=0, \; 
\Sigma_m(\mb{u},p)=\nu(\nabla\mb{u}+(\nabla\mb{u})^T)-pI
\end{equation}
where $I$ is the $3\times 3$ identity matrix and $(\nabla\mb{u})^T$
is the transpose of $\nabla\mb{u}$.
Here, $\Sigma_m$ is the mechanical stress tensor.
It is possible to carry out much of the calculations to follow
even if we retain inertial terms and work with the Navier-Stokes equations 
or use other constitutive relations for the mechanical stress.
In particular, such modifications will not destroy 
the free energy identity to be discussed below.
We do note, however, that incompressibility is important for our computations.

We now turn to boundary conditions. We let $\mb{u}=0$ on the 
outer boundary $\Gamma_\text{out}$ for simplicity. On the cell membrane $\Gamma$, 
we have the following conditions. 
Take a point $\mb{x}=\mb{X}(\bm{\theta},t)$ on the boundary $\Gamma$, 
and let $\mb{n}$ be the unit outward normal on $\Gamma$ at this point.
First, by force balance, we have:
\begin{equation}
[\Sigma_m(\mb{u},p)]=\mb{F}_\text{elas}.\label{forcebalance}
\end{equation}
Here, $\mb{F}_\text{elas}$ is the 
elastic force per unit area of membrane.

We make some assumptions about the form of the elastic force.
We assume that the membrane is a hyperelastic material in the sense that 
the elastic force can be derived from an elastic 
energy functional $E_\text{elas}$ that is a function only 
of the configuration $\mb{X}$:
\begin{equation}
E_\text{elas}(\mb{X})=\int_{\Gamma_\text{ref}} \mathcal{E}(\mb{X})dm_{\Gamma_\text{ref}}
\label{elas}
\end{equation}
where $m_{\Gamma_\text{ref}}$ is the surface 
measure of $\Gamma_\text{ref}$ and 
$\mathcal{E}$ is the elastic energy density measured with respect 
to this measure. It is possible that $\mathcal{E}$ is a function of 
spatial derivatives of $\mb{X}$.
The elastic force $\mb{F}_\text{elas}(\mb{x})$ satisfies the relation:
\begin{equation}
\begin{split}
\at{\D{}{s}}{s=0}\int_{\Gamma_\text{ref}} 
\mathcal{E}(\mb{X}(\bm{\theta})+s\mb{Y}(\bm{\theta}))
dm_{\Gamma_\text{ref}}
&=-\int_{\Gamma} \mb{F}_\text{elas}(\mb{x})\cdot 
\mb{Y}(\mb{X}^{-1}(\mb{x}))dm_\Gamma
\end{split}
\end{equation}
where $\mb{Y}$ is an arbitrary vector field defined on $\Gamma_\text{ref}$
and $m_\Gamma$ is the natural measure on the surface $\Gamma$
and is related to $m_{\Gamma_\text{ref}}$ 
by $dm_\Gamma=Qdm_{\Gamma_\text{ref}}$ where $Q$ is the Jacobian determinant 
relating $\Gamma_t$ to the reference configuration $\Gamma_\text{ref}$. 
The expression $\mb{X}^{-1}(\mb{x})$
is the inverse of the map $\mb{x}=\mb{X}(\bm{\theta})$.
Thus, $\mb{F}_\text{elas}$ is given as 
the variational derivative of the elastic 
energy up to the Jacobian factor $Q$.
Consequently, we have the following relation:
\begin{equation}
\D{}{t}E_\text{elas}(\mb{X})
=-\int_{\Gamma} \mb{F}_\text{elas}\cdot \PD{\mb{X}}{t}dm_\Gamma.\label{dtelas}
\end{equation}
In the above, $\PD{\mb{X}}{t}$ should be thought of as a function of 
$\mb{x}$, i.e., 
$\PD{\mb{X}}{t}=\PD{\mb{X}}{t}(\mb{X}^{-1}(\mb{x}))$. 
We shall henceforth 
abuse notation and let $\PD{\mb{X}}{t}$ be a function of $\mb{x}$
or $\bm{\theta}$ depending on the context of the expression.

In addition to the force balance condition \eqref{forcebalance}, 
we need a continuity 
condition on the interface $\Gamma$. Since we are allowing for 
osmotic water flow, we have a slip between the 
movement of the membrane and the flow field. 
At a point $\mb{x}=\mb{X}(\bm{\theta},t)$ on the boundary $\Gamma$
we have:
\begin{equation}
\mb{u}-\PD{\mb{X}}{t}=j_w\mb{n}\label{cont}
\end{equation}
where $j_w$ is water flux through the membrane. We are thus assuming 
that water flow is always normal to the membrane and 
that there is no slip between the fluid and the membrane
in the direction tangent to the membrane. Given that 
$\mb{n}$ is the outward normal, $j_w$ is positive when 
water is flowing out of the cell. Like $j_k$ in \eqref{ckbc}, 
we let $j_w$ be a passive flux in the following sense.
We let:
\begin{equation}\label{jwexp}
j_w=j_w([\widehat{\psi_w}]),\; 
\at{\widehat{\psi_w}}{\Gamma_{i,e}}=
\at{\pi_w}{\Gamma_{i,e}}
-\at{((\Sigma_m(\mb{u},p))\cdot\mb{n})}{\Gamma_{i,e}}
\end{equation}
where $\pi_w$ was the entropic contribution to the water 
potential defined in \eqref{muw}. In general, 
$j_w$ can be a function of other variables (see Section \ref{cross}). 
The function $j_w$ satisfies the condition, analogous to \eqref{jkcond}:
\begin{equation}\label{jwcond}
[\widehat{\psi_w}]j_w\geq 0.
\end{equation}
This condition is clearly satisfied if:
\begin{equation}\label{jwmonotone}
j_w=j_w([\widehat{\psi_w}]),\; j_w(0)=0,\; \PD{j_w}{[\widehat{\psi_w}]}>0.
\end{equation}
Water flow across the membrane is thus driven by the difference 
in the entropic contribution to the chemical potential as well 
as the jump in the mechanical force across the cell membrane. 
When the flow 
field $\mb{u}$ is equal to $0$, the above expression for $\widehat{\psi_w}$
reduces to:
\begin{equation}
\at{\widehat{\psi_w}}{\Gamma_{i,e}}=
\at{\pi_w}{\Gamma_{i,e}}
+\at{p}{\Gamma_{i,e}}=\at{\psi_w}{\Gamma_{i,e}}.
\end{equation}
We may thus view $\widehat{\psi_w}$ as a modification 
of $\psi_w$ to take into account dynamic flow effects.
Let $\omega=\omega_0$ given in \eqref{ent}
in the definition \eqref{muw} of $\pi_w$.
Then, we have, under zero flow conditions,
\begin{equation}
[\widehat{\psi_w}]=\jump{p-k_BT\sum_{k=1}^N c_k}.\label{muwclassical}
\end{equation}
We thus reproduce the standard statement that water flow across the membrane
is driven by the difference in osmotic and mechanical pressure, 
where the osmotic pressure, $\pi_w$, is given by the van t'Hoff law.

\subsection{Free Energy Identity} 

We now show that the system described above satisfies the 
following free energy identity.

\begin{theorem}\label{mainc}
Suppose $c_k,\mb{u}, p$ be smooth functions 
that satisfy \eqref{ckeq}, \eqref{stokes},
and in $\Omega_i$ and $\Omega_e$ and satisfy boundary 
conditions \eqref{ckbc}, \eqref{forcebalance}, \eqref{cont}
on the membrane $\Gamma$. Suppose further
that $c_k$ and satisfy no-flux boundary conditions and 
$\mb{u}=0$ on the outer boundary $\Gamma_\text{out}$. Then, $c_k,\mb{u},p$
and $\phi$ satisfy the following free energy identity.
\begin{equation}
\begin{split}
\D{}{t}(G_S+E_\text{elas})&=-I_p-J_p-J_a\\
G_S&=\int_{\Omega_i\cup \Omega_e}\omega d\mb{x}\\
I_p&=\int_{\Omega_i\cup \Omega_e}\paren{\nu\abs{\nabla \mb{u}}^2
+\sum_{k=1}^Nc_k\frac{D_k}{k_BT}\abs{\nabla \mu_k}^2}d\mb{x}\\
J_p&=\int_{\Gamma}
\paren{[\widehat{\psi_w}]j_w
+\sum_{k=1}^N[\mu_k]j_k}dm_\Gamma\\
J_a&=\int_{\Gamma}\paren{\sum_{k=1}^N[\mu_k]a_k}dm_\Gamma.
\end{split}\label{FE}
\end{equation}
Here, $E_\text{elas}$ was given in \eqref{elas} and
$\abs{\nabla \mb{u}}$ is the Frobenius norm of the $3\times 3$ rate of 
deformation matrix $\nabla \mb{u}$. 

If $a_k\equiv 0$, then the free energy is monotone decreasing.
\end{theorem}
The free energy is given as a sum 
of the entropic contribution $G_S$ and
the membrane elasticity term $E_\text{elas}$.
This free energy is dissipated through bulk currents 
$I_p$ and membrane currents $J_p$. Dissipation in the bulk comes from
ionic electrodiffusion and viscous dissipation. Dissipation at the membrane 
comes from ionic channel currents and transmembrane water flow.  
If active membrane currents are present, they may contribute 
to an increase in the free energy through the term $J_a$.
It is sometimes useful to rewrite $J_p+I_p$ as:
\begin{equation}
\begin{split}\label{Fwc}
J_p+I_p&=F_w+F_c,\\
F_w&=\int_{\Omega_i\cup \Omega_e}\nu\abs{\nabla \mb{u}}^2d\mb{x}
+\int_{\Gamma}[\widehat{\psi_w}]j_wdm_\Gamma,\\
F_c&=\int_{\Omega_i\cup \Omega_e}
\sum_{k=1}^Nc_k\frac{D_k}{k_BT}\abs{\nabla \mu_k}^2d\mb{x}
+\int_{\Gamma}\sum_{k=1}^N[\mu_k]j_kdm_\Gamma,
\end{split}
\end{equation}
where $F_w$ and $F_c$ are the dissipations due to water flow 
and solute diffusion respectively. 
In the statement of the Theorem, it is important that
\eqref{jkcond} and \eqref{jwcond} are used only to conclude 
that $J_p$ be positive. Identity \eqref{FE} should be seen 
as giving us the definition of what a passive current should be.

We now prove Theorem \ref{main}.
An interesting point about 
the calculation to follow is how dissipation through 
transmembrane water flow comes from two different sources, equations for 
ionic concentration dynamics and the fluid equations.
The former contributes the osmotic term $\pi_w$ term,
and the latter contributes the mechanical term $p$, which together 
add up to the water potential $\psi_w$. 

\begin{proof}[Proof of Theorem \ref{mainc}]
First, multiply \eqref{ckeq} with $\mu_k$ in \eqref{muk} and integrate 
over $\Omega_i$ and sum in $k$:
\begin{equation}
\sum_{k=1}^N\int_{\Omega_i}\mu_k\paren{\PD{c_k}{t}+\nabla\cdot (\mb{u} c_k)}d\mb{x}
=\sum_{k=1}^N\int_{\Omega_i}\mu_k \nabla \cdot \paren{c_k\frac{D_k}{k_BT}\nabla \mu_k}d\mb{x}.\label{mkck}
\end{equation}
The summand in the right hand side becomes:
\begin{equation}
\begin{split}
\int_{\Omega_i}\mu_k \nabla \cdot \paren{c_k\frac{D_k}{k_BT}\nabla \mu_k}d\mb{x}
&=\int_{\Gamma_i}\paren{\mu_k {c_k\frac{D_k}{k_BT}\nabla \mu_k}\cdot \mb{n}}dm_\Gamma\\
&-\int_{\Omega_i}\paren{c_k\frac{D_k}{k_BT}\abs{\nabla \mu_k}^2}d\mb{x}
\end{split}
\end{equation}
where $\mb{n}$ is the outward normal on $\Gamma$. Consider the 
left hand side of \eqref{mkck}.
\begin{equation}
\sum_{k=1}^N\mu_k\paren{\PD{c_k}{t}+\nabla \cdot (\mb{u} c_k)}
=\sum_{k=1}^N\PD{\omega}{c_k}\paren{\PD{c_k}{t}+\mb{u}\cdot \nabla c_k}
=\PD{\omega}{t}+\nabla \cdot (\mb{u}\omega),\label{mkcksum}
\end{equation}
where we used \eqref{mukc} and the incompressibility condition 
in \eqref{stokes}.
Integrating the above over $\Omega_i$, we have:
\begin{equation}
\begin{split}
\int_{\Omega_i} 
\paren{\PD{\omega}{t}+\nabla \cdot (\mb{u}\omega)}d\mb{x}
&=\int_{\Omega_i} \PD{\omega}{t}d\mb{x}
+\int_{\Gamma_i} \omega\mb{u}\cdot\mb{n}dm_\Gamma\\
&=\D{}{t}\int_{\Omega_i}\omega d\mb{x}
+\int_{\Gamma_i}\omega\paren{\mb{u}-\PD{\mb{X}}{t}}\cdot\mb{n}dm_\Gamma
\end{split}
\end{equation}
where we used the fact that 
$\mb{u}$ is divergence free in the first equality.
The term involving $\PD{\mb{X}}{t}$ comes from the fact that the 
membrane $\Gamma$ is moving in time.
Performing similar calculations on $\Omega_e$, and adding this to 
the above, we find:
\begin{equation}\label{domega}
\begin{split}
&\D{}{t}\int_{\Omega_i\cup \Omega_e}\omega d\mb{x}
+\int_{\Gamma}[\omega]j_w dm_\Gamma\\
=&\sum_{k=1}^N
\int_{\Gamma}\jump{\mu_k {c_k\frac{D_k}{k_BT}\nabla \mu_k}\cdot \mb{n}}dm_\Gamma
-\sum_{k=1}^N
\int_{\Omega_i\cup \Omega_e}\paren{c_k\frac{D_k}{k_BT}\abs{\nabla \mu_k}^2}d\mb{x}.
\end{split}
\end{equation}
where we used \eqref{cont}. 
Using \eqref{ckbc} and \eqref{cont}, 
we may rewrite the second boundary integral as follows:
\begin{equation}\label{concbndry}
\int_{\Gamma}\jump{\mu_k {c_k\frac{D_k}{k_BT}\nabla \mu_k}
\cdot \mb{n}}dm_\Gamma
=\int_{\Gamma}\paren{\jump{\mu_k c_k}j_w
-[\mu_k](j_k+a_k)}dm_\Gamma.
\end{equation}
We now turn to equation \eqref{stokes}. Multiply this by $\mb{u}$
and integrate over $\Omega_i$:
\begin{equation}
\int_{\Omega_i} \mb{u}\cdot(\nu \Delta \mb{u} -\nabla p)d\mb{x}
=\int_{\Gamma_i}\paren{\Sigma_m(\mb{u},p)\mb{n}}\cdot\mb{u}dm_\Gamma-
\int_{\Omega_i}\mb{\nu}\abs{\nabla \mb{u}}^2d\mb{x}=0
\end{equation}
Performing a similar calculation on $\Omega_e$ and adding this to the above, 
we have:
\begin{equation}
\int_{\Gamma}\jump{\paren{\Sigma_m(\mb{u},p)\mb{n}}}\cdot \mb{u}dm_\Gamma-
\int_{\Omega_i\cup \Omega_e}\mb{\nu}\abs{\nabla \mb{u}}^2d\mb{x}
=0
\end{equation}
We may use \eqref{forcebalance}, \eqref{dtelas} and \eqref{cont} to find
\begin{equation}\label{elasen}
\D{}{t} E_\text{elas}(\mb{X})=
\int_{\Gamma}\jump{\paren{\Sigma_m(\mb{u},p)\mb{n}}\cdot \mb{n}}j_wdm_\Gamma
-\int_{\Omega_i\cup \Omega_e}\mb{\nu}\abs{\nabla \mb{u}}^2d\mb{x}
\end{equation}
Combining \eqref{domega}, \eqref{concbndry} and \eqref{elasen}, 
we have:
\begin{equation}
\begin{split}
&\D{}{t}\paren{\int_{\Omega_i\cup \Omega_e}\omega d\mb{x}+E_\text{elas}(\mb{X})}\\
=&-\sum_{k=1}^N\int_{\Omega_i\cup \Omega_e}\paren{c_k\frac{D_k}{k_BT}\abs{\nabla \mu_k}^2}d\mb{x}-\int_{\Omega_i\cup \Omega_e}\mb{\nu}\abs{\nabla \mb{u}}^2d\mb{x}\\
&-\int_{\Gamma}[\mu_k](j_k+a_k)dm_\Gamma
-\int_{\Gamma}\jump{\omega-c_k\sigma_k-(\Sigma_m(\mb{u},p)\mb{n})\cdot\mb{n}}j_wdm_\Gamma
\end{split}
\end{equation}
Recalling the definition of $\widehat{\psi_w}$ in \eqref{jwexp}, we obtain 
the desired equality.
In the absence of active currents $a_k$, is decreasing 
given that $j_k$ and $j_w$ satisfy conditions 
\eqref{jkcond} and \eqref{jwcond} respectively.
\end{proof}

In the last line of the above proof, note that the expression 
for $\widehat{\psi_w}$ arises 
naturally as a result of integrating by parts. In this sense, 
we may say that osmotic water flow arises as a natural consequence 
of requiring that the free energy be decreasing in time.  

\subsection{Cross Coefficients and Solvent Drag}\label{cross}

As can be seen from \eqref{FE} or \eqref{FEE}
the only condition we need to impose for the free energy 
to decrease with time in the absence of active currents is the 
following:
\begin{equation}\label{jkjw}
[\widehat{\psi_w}]j_w+\sum_{k=1}^N [\mu_k]j_k\geq 0.
\end{equation}
This condition is weaker than conditions \eqref{jkcond} and \eqref{jwcond}
being satisfied separately by $j_k$ and $j_w$. We now discuss an 
important case in which $j_k$ and $j_w$ may not individually satisfy 
\eqref{jkcond} and \eqref{jwcond} but \eqref{jkjw} is satisfied nevertheless.
This arises whenever fluxes are coupled as is usually 
the case for fluxes through transporters or (single filing) channels
\cite{hille1982transport,Hille,tosteson1989membrane,boron2008medical,
davson1970textbook}. We note that such cross-diffusion can be 
relevant even in bulk solution
\cite{tyrrell1971diffusion,justice1983conductance,hoheisel1993theoretical,
taylor1993multicomponent,accascina1959electrolytic}.

If $\jump{\mu_k}$ and $[\widehat{\psi_w}]$ remain 
small, the dissipation $J_p$ in \eqref{FE} may be approximated 
by a quadratic form in the jumps:
\begin{equation}\label{quad}
\begin{split}
J_p&=\int_\Gamma \jump{\bm{\mu}}\cdot\mb{j}dm_\Gamma=\int_\Gamma \jump{\bm{\mu}}\cdot(\mathcal{L}\jump{\bm{\mu}})dm_\Gamma,\\
\bm{\mu}&=(\mu_1,\cdots,\mu_N,\widehat{\psi_w})^T, 
\mb{j}=(j_1,\cdots,j_N,j_w)^T,
\end{split}
\end{equation}
where $\mathcal{L}$ is a symmetric $(N+1)\times(N+1)$ matrix.
Requiring that the free energy be decreasing implies that $\mathcal{L}$
must be positive definite. 
The {\em maximum dissipation principle} requires that $\mb{j}$ be given as 
variational derivatives of $J_p/2$ with respect to $[\bm{\mu}]$:
\begin{equation}
\mb{j}=\mathcal{L}\jump{\bm{\mu}}.
\end{equation}
Note that, without the maximum dissipation principle, 
\eqref{quad} only implies 
$\mb{j}=(\mathcal{L}+\tilde{\mathcal{L}})\jump{\bm{\mu}}$ 
where $\tilde{\mathcal{L}}$
is an arbitrary skew symmetric matrix.
The symmetry of the coefficient matrix $\mathcal{L}$ relating $[\bm{\mu}]$
and $\mb{j}$ is an instance of the Onsager reciprocity 
relation \cite{degroot1962non,katzir1965nonequilibrium,kjelstrup2008non}. 

A lipid bilayer membrane is impermeable to many solutes, but only 
approximately so. In this case, a water flux
may induce a solute flux, and this may be expressed as 
$\mathcal{L}_{kw}\neq 0$ where $\mathcal{L}_{kw}$ is the $(k,N+1)$ entry 
of the matrix $\mathcal{L}$. This is known as solvent drag.
Given the presence of such cross coefficients, \eqref{jkcond}
and \eqref{jwcond} are not necessary true, 
whereas condition \eqref{jkjw} is true by construction.

\section{Electrodiffusion of Ions and Osmotic Water Flow}\label{elecdiff}

\subsection{Model Formulation}
Let us now consider the case in which the chemical species 
are electrically charged.  
As in the previous section, we let $c_k, k=1,\cdots, N$
be the concentrations of the ionic species. 
Given $\omega$, the entropic part of the free energy per unit volume, 
the chemical potential $\mu_k$ 
of the $k-$th species of ion is given as:
\begin{equation}
\mu_k=\PD{\omega}{c_k}+qz_k\phi= \sigma_k+qz_k\phi\label{muk}.
\end{equation}
The chemical potential is thus a sum of the entropic term $\sigma_k$ and 
the electrostatic term. 
In the electrostatic term, $q$ is the elementary charge, $z_k$ is the valence
of the $k$-th species of ion, and $\phi$ is the electrostatic potential.
The definition of the water potential, $\psi_w$ and 
$\widehat{\psi_w}$, remain the same. 
The ionic concentrations $c_k$ satisfy \eqref{ckeq} 
and \eqref{ckbc} except that 
we now use \eqref{muk} as our expression for the chemical potential.
Ions are thus subject to drift by the electric field in addition 
to diffusion and advection by the local flow field.

If the electrolyte solution is sufficiently dilute, the chemical 
potential $\mu_k$ is given by \eqref{muk} with $\omega$ 
equal to \eqref{ent}. However, 
deviations from ideality can be significant in
electrolyte solutions, especially in higher concentrations
\cite{fawcett2004liquids,lee2008molecular,kunz2010specific,fraenkel2010simplified,
eisenberg2010crowded}.
Cross-diffusion (or flux coupling)
in the bulk can also be significant in electrolyte solutions 
\cite{tyrrell1971diffusion,justice1983conductance,hoheisel1993theoretical,taylor1993multicomponent,accascina1959electrolytic}.
These effects are clearly important in describing the {\em molecular} physiology
of ion channel pores and enzyme active sites at which ionic concentrations 
can reach tens of molars \cite{eisenberg2010crowded,zhang2010molecular}.
The question of whether these effects are significant in formulating 
phenomenological models in 
{\em cellular} physiology, where the typical ionic concentrations 
are two orders of magnitude lower, is largely unexplored.  
This exploration is beyond the scope of the 
present paper, but we point out that our formalism allows the  
incorporation of such effects \cite{eisenberg2010energy}. 

The transmembrane flux $j_k$ is now a function of the 
membrane potential $[\phi]$ in addition to the dependencies 
discussed in the previous section. 
We require that $j_k$ satisfy condition \eqref{jkcond}.

The electrostatic potential $\phi$ 
satisfies the Poisson equation:
\begin{equation}
-\nabla \cdot (\epsilon\nabla\phi)=\sum_{k=1}^N qz_kc_k\label{poisson}
\end{equation}
where $\epsilon$ is the 
dielectric constant. We shall assume that $\epsilon$ is constant 
in space and time. This restriction may be lifted, at the expense
of introducing a relation that describes the evolution of $\epsilon$.
We also assume that there is no fixed background charge. 
It is easy to generalize the calculations below to the 
case when the immobile charges, if present, always stay away 
from the moving membrane. Otherwise, one would need to introduce 
``collision rules'' to determine what happens when the membrane 
hits the immobile charges.
We impose Neumann boundary conditions for \eqref{poisson} on 
the outer boundary $\Gamma_\text{out}$ for simplicity. On the membrane 
$\Gamma$, we impose the following boundary condition:
\begin{equation}
-\at{\epsilon\PD{\phi}{\mb{n}}}{\Gamma_i}
=-\at{\epsilon\PD{\phi}{\mb{n}}}{\Gamma_e}=C_m[\phi].\label{cap}
\end{equation}
where $C_m$ is the capacitance per unit area of membrane. 
The above is simply a statement about the continuity 
of the electric flux density.
Since the membrane is moving, the capacitance $C_m$ is itself 
an evolving quantity. 
We assume the following family of constitutive laws for $C_m$. 
At $\mb{x}=\mb{X}(\bm{\theta},t)$,
\begin{equation}
C_m(\mb{x})=C_m(Q(\mb{X}))\label{CmQ}
\end{equation}
where $Q(\mb{X})$ is the Jacobian or metric determinant of 
the configuration $\Gamma_t$ at time $t$
with respect to the reference configuration $\Gamma_\text{ref}$.
This factor describes the extent to which the membrane 
is stretched from the rest configuration. 
A simple example of \eqref{CmQ} would be:
\begin{equation}
C_m(\mb{x})=C_m^0=\text{const}.\label{Cm0}
\end{equation}
As another example, we may set:
\begin{equation}
C_m(\mb{x})=C_m^0 Q(\mb{X})\label{Cmscaling}
\end{equation}
where $C_m^0=\text{const}$ 
is the capacitance per unit area measured in the 
reference configuration. 
Relation \eqref{Cmscaling} is the natural scaling if we assume 
that the membrane is made of an incompressible material. 
For suppose the membrane is made of a material whose dielectric 
constant is $\epsilon_m$. If the thickness of the membrane at 
the point $\mb{x}=\mb{X}(\bm{\theta},t)$ is $d(\mb{x})$, 
the membrane capacitance there 
is given by $\epsilon_m/d(\mb{x})$. The incompressibility 
of the material implies that the local membrane volume 
remains constant in time: 
$d(\mb{x})Q(\mb{X})=\text{const}$. Thus, 
$C_m(\mb{x})$ must be proportional to $Q(\mb{X})$.

Force balance must be modified to take into account 
electrostatic forces.The flow field 
$\mb{u}$ satisfies the Stokes equation 
in $\Omega_i$ or $\Omega_e$ with an electrostatic force term:
\begin{equation}
\nu\Delta \mb{u}-\nabla p-\paren{\sum_{k=1}^Nqz_kc_k}\nabla \phi
=0, \; \nabla \cdot \mb{u}=0 \label{stokesE}
\end{equation}
Note that the above equations can also be written as follows:
\begin{equation}
\begin{split}
\nabla \cdot (\Sigma_m(\mb{u},p)+\Sigma_e(\phi))&=0,\; \nabla \cdot \mb{u}=0,\\
\Sigma_m(\mb{u},p)&=\nu(\nabla{u}+(\nabla{u})^T)-pI,\\
\Sigma_e(\phi)&= 
\epsilon\paren{\nabla \phi\otimes \nabla \phi-\frac{1}{2}\abs{\nabla \phi}^2I}.
\end{split}
\end{equation}
Here, $\Sigma_e$ is the 
Maxwell stress tensor generated by the electric field. 
Note that we have used \eqref{poisson} to rewrite the electrostatic force in 
\eqref{stokesE} in terms of $\Sigma_e$.

We now turn to boundary conditions. We continue to let $\mb{u}=0$ on the 
outer boundary $\Gamma_\text{out}$. On the cell membrane $\Gamma$, 
we have the following conditions. 
First, by force balance, we have:
\begin{equation}
[(\Sigma_m(\mb{u},p)+\Sigma_e(\phi))\mb{n}]
=\mb{F}_\text{elas}+\mb{F}_\text{cap}\label{forcebalanceE}
\end{equation}
In addition to $\mb{F}_\text{elas}$, we have an additional term
$\mb{F}_\text{cap}$ which arises because the membrane 
carries capacitive energy. We shall call this the capacitive force, 
which is given as:
\begin{equation}
\mb{F}_\text{cap}=\tau_\text{cap}\kappa_\Gamma\mb{n}
-\nabla_\Gamma \tau_\text{cap}, \; \tau_\text{cap}
=\frac{1}{2}\paren{C_m+Q\PD{C_m}{Q}}[\phi]^2
\label{FcapCmQ}
\end{equation}
where $\kappa_\Gamma$ is the sum of the principal curvatures of the 
membrane $\Gamma$
and $\nabla_\Gamma=\nabla-\mb{n}(\mb{n}\cdot \nabla)$ is the surface 
gradient on $\Gamma$.
The above expression shows that the capacitive force can 
be seen as a surface tension of strength $-\tau_\text{cap}$.
The above capacitive force is chosen so that Theorem \ref{main} holds,
and in this sense, the proof of Theorem \ref{main} provides a variational 
interpretation of this force. In Appendix \ref{capforce}, we give a 
physical interpretation of expression \eqref{FcapCmQ}.  

An interesting variant of \eqref{forcebalanceE} is the following. 
Suppose the membrane is incompressible in the sense that $Q\equiv 1$
for all time. We note that this condition of two-dimensional 
incompressibility is {\em not} the same 
as assuming that the membrane is made of a (three-dimensional) 
incompressible material. In the case of three-dimensional incompressibility, 
the membrane may stretch, 
but this would lead to a thinning of the membrane, leading 
to the constitutive law \eqref{Cmscaling} as we saw earlier.
When $Q\equiv 1$ for all time, we let:
\begin{equation}
[(\Sigma_m(\mb{u},p)+\Sigma_e(\phi))\mb{n}]
=\mb{F}_\text{elas}+\mb{F}_\text{p}\label{forcebalancep}
\end{equation}
where $\mb{F}_p$ is given as:
\begin{equation}\label{Fp}
\mb{F}_\text{p}=\lambda \kappa_\Gamma\mb{n}-\nabla_\Gamma \lambda.
\end{equation}
The above is a surface pressure and 
$\lambda$ is determined so that $Q\equiv 1$.
Note that in \eqref{forcebalancep} we do not need a capacitive 
force since it can be absorbed into the surface pressure term.

The continuity condition \eqref{cont} remains the same. 
We continue to require that the passive flux $j_k$ satisfy 
\eqref{jkcond}. An important difference, however, is that 
$j_k$ now depends strongly on the membrane potential $[\phi]$,
given that $\mu_k$ depends on $\phi$. Ions usually flow 
through ionic channels, which often have open and closed 
states. The passive flux $j_k$ may also depend on such 
states, which are described by gating variables. 
This flux is the subject of the large experimental 
and theoretical work on ion channels \cite{Hille}. 
For the present work, 
it is appropriate to use the classical phenomenological treatments of flux,
although their molecular underpinnings are not clear \cite{chen1997permeation,
eisenberg1999structure,gillespie2002physical}. 
Some of popular choices $j_k$ include \cite{KS,Hille,HH}:
\begin{align}
j_k^\text{HH}&=g_k[\mu_k]=g_k\paren{z_k[\phi]
+\ln\paren{\frac{\at{c_k}{\Gamma_i}}{\at{c_k}{\Gamma_e}}}},\label{HH}\\
j_k^\text{GHK}&=P_kz_k\phi'\paren{
\frac{\at{c_k}{\Gamma_i}\exp(z_k\phi')-\at{c_k}{\Gamma_e}}
{\exp(z_k\phi')-1}}, \; \phi'=\frac{q\jump{\phi}}{k_BT},\label{GHK}
\end{align}
where $g_k$ and $P_k$ are positive and depend 
on the gating variables in certain modeling contexts. 
It is easily seen that both $j_k^\text{HH}$ and $j_k^\text{GHK}$
satisfy \eqref{jkcond}. We shall use expression \eqref{GHK}
in our numerical computations in Section \ref{animal}.

We remark that the model we just proposed is nothing other 
than the Poisson-Nernst-Planck-Stokes system if we let $\omega=\omega_0$
given in \eqref{ent} 
\cite{Rubinstein}. The novelty here 
is in the interface conditions at the membrane, \eqref{ckbc},
\eqref{cap}, \eqref{forcebalanceE} and \eqref{cont}.
The Poisson Nernst-Planck system has received much attention 
in the field of semiconductors \cite{roosbroeck_theory_1950,Jerome_semiconductor,
selberherr1984analysis}, ionic channels \cite{eisenberg1996computing},
ion exchange membranes and desalination \cite{Rubinstein} as well as 
physical chemistry \cite{bazant2004diffuse}. 

\subsection{Free Energy Identity}

We now show that the system described in the previous section
possesses a natural free energy.

\begin{theorem}\label{main}
Suppose $c_k,\mb{u},p$ and $\phi$ are smooth functions 
that satisfy \eqref{ckeq}, \eqref{stokesE},
and \eqref{poisson} in $\Omega_i$ and $\Omega_e$ and satisfy boundary 
conditions \eqref{ckbc}, \eqref{forcebalanceE}, \eqref{cont}
and \eqref{cap} on the membrane $\Gamma$. Suppose further
that $c_k$ and $\phi$ satisfy no-flux boundary conditions and 
$\mb{u}=0$ on the outer boundary $\Gamma_\text{out}$. Then, $c_k,\mb{u},p$
and $\phi$ satisfy the following free energy identity.
\begin{equation}
\begin{split}
\D{}{t}(G_S+E_\text{elas}+E_\text{elec})&=-I_p-J_p-J_a\\
E_\text{elec}&=\int_{\Omega_e\cup \Omega_e}\frac{1}{2}\epsilon \abs{\nabla{\phi}}^2d\mb{x}
+\int_{\Gamma}\frac{1}{2}C_m[\phi]^2dm_\Gamma
\end{split}\label{FEE}
\end{equation}
Here, $G_S, E_\text{elas}, I_p, J_p, I_a$ are the same as in \eqref{FE}.
The same identity holds if we require $Q\equiv 1$ and 
adopt \eqref{forcebalancep} instead of \eqref{forcebalanceE}.

If $a_k\equiv 0$, the free energy is monotone decreasing.
\end{theorem}
In addition to the terms present in \eqref{FEE}, we now have 
an electrostatic term in the energy.

Before proving Theorem \ref{main}, we collect some calculus results. 
We first introduce some notation. 
Take a point $\mb{x}=\mb{x}_0\in \Gamma$ at $t=t_0$. 
Let $\mb{X}^\mb{n}(t;x_0,t_0)$ be 
the space-time curve that 
goes through $\mb{x}=\mb{x}_0$ at time $t=t_0$ and is orthogonal 
to $\Gamma$ at each time instant.
Equivalently, $\mb{X}^\mb{n}(t;\mb{x}_0,t_0)$ is the solution to the 
following ordinary differential equation:
\begin{equation}
\D{}{t}\mb{X}^\mb{n}(t;\mb{x}_0,t_0)
=v_\Gamma(\mb{X}^\mb{n},t)\mb{n}(\mb{X}^\mb{n},t),
\; \mb{X}^\mb{n}(t_0;,\mb{x}_0,t_0)=\mb{x}_0.
\end{equation}
Here, $\mb{n}(\mb{x},t)$ is the unit normal at the point 
$\mb{x}$ at time $t$ pointing from $\Omega_i$ into $\Omega_e$,
and $v_\Gamma(\mb{x},t)\mb{n}(\mb{x},t)$ is the normal velocity of $\Gamma$
at that point.
Consider a quantity $w(\mb{x},t)$ defined 
on the evolving surface $\Gamma$. Define:
\begin{equation}
(D^\mb{n}_t w)(\mb{x}_0,t_0)=
\at{\D{}{t}w(\mb{X}^\mb{n}(t;\mb{x}_0,t_0),t)}{\mb{x}=\mb{x}_0,t=t_0}.
\label{Dnt}
\end{equation} 
The above expression is an analogue of the convective derivative 
on the surface $\Gamma$. We shall make use of the following 
well-known identity:
\begin{equation}
\D{}{t}\int_\Gamma wdm_\Gamma=
\int_\Gamma \paren{D^\mb{n}_t w+\kappa_\Gamma w v_\Gamma}dm_\Gamma
\label{Dtw}
\end{equation}
where $\kappa_\Gamma$ is the sum of the principal curvatures of $\Gamma$.
We now state two calculus identities that we shall find useful in 
the proof of Theorem \ref{main}. 

\begin{lemma}\label{divongamma}
Let $w(\mb{x},t)$ be a smooth function on $\Gamma_t$.
We have:
\begin{equation}
\int_\Gamma \paren{wQ^{-1}\PD{Q}{t}}dm_\Gamma=
\int_\Gamma \paren{\kappa_\Gamma w \mb{n}-(\nabla_\Gamma w)}
\cdot\PD{\mb{X}}{t}dm_\Gamma.\label{Qt}
\end{equation}
where
$Q$ is the Jacobian determinant of $\Gamma_t$ with respect 
to the reference configuration $\Gamma_\text{ref}$.
\end{lemma}
\begin{proof}
Note that
\begin{equation}
\PD{w}{t}=D_t^\mb{n}w+(\nabla_\Gamma w)\cdot\PD{\mb{X}}{t}\label{wt}
\end{equation}
where the partial derivatives in $t$ is along material trajectories
(constant $\bm{\theta}$). The validity of the above identity should be 
clear by considering the geometric relation between the orthogonal trajectory 
$\mb{X}^\mb{n}$ and the material trajectory $\mb{X}$. 
We also have the following relation for the time derivative of the integral 
of $w$ over $\Gamma$.
\begin{equation}
\begin{split}
&\D{}{t}\int_\Gamma w dm_\Gamma=\D{}{t}\int_{\Gamma_{\text{ref}}}
wQdm_{\Gamma_\text{ref}}
=\int_{\Gamma_{\text{ref}}}\paren{\PD{w}{t}Q+w\PD{Q}{t}}dm_{\Gamma_\text{ref}}\\
=&\int_\Gamma\paren{\PD{w}{t}+wQ^{-1}\PD{Q}{t}}dm_\Gamma
\end{split}
\end{equation}
Comparing this with \eqref{Dtw} (with $v_\Gamma=\PD{\mb{X}}{t}\cdot\mb{n}$) 
and using the identity \eqref{wt}, we obtain the desired result.
\end{proof}

\begin{lemma}\label{calcidentity}
Suppose $w(\mb{x},t), \mb{x}\in (\Omega_i\cup \Gamma)$ 
is a smooth function defined in $\Omega_i$ whose
derivatives are continuous up to the 
boundary $\Gamma$. Then, we have the following identity:
\begin{equation}
\begin{split}
&\int_{\Gamma_i}\paren{w\PD{}{\mb{n}}\paren{\PD{w}{t}}+(w\Delta w) v_\Gamma}
dm_\Gamma\\
=&\int_{\Gamma_i}\paren{w D_t^\mb{n}\paren{\PD{w}{\mb{n}}}+
\paren{\kappa_\Gamma w\PD{w}{\mb{n}}-\abs{\nabla_\Gamma w}^2} v_\Gamma}dm_\Gamma \label{identity}
\end{split}
\end{equation}
where $\int_{\Gamma_i}$ denotes 
integration over the $\Omega_i$ face of $\Gamma$.
A similar identity holds for functions defined in $\Omega_e\cup \Gamma$.
\end{lemma}
As we shall see, we only need $w$ to be defined in the vicinity 
of $\Gamma$ for the above to be true.
\begin{proof} 
We only treat the $\Gamma_i$ case.
The proof for $\Gamma_e$ is exactly the same.
We decompose the integrand in the left hand side of \eqref{identity}
into tangential and normal contributions.
It is well-known that the Laplacian can be written as:
\begin{equation}
\Delta w=\PDD{2}{w}{\mb{n}}+\kappa_\Gamma \PD{w}{\mb{n}}+\Delta_\Gamma w
\label{lapphi}
\end{equation}
where $\Delta_\Gamma$ is the Laplace-Beltrami 
operator of the surface $\Gamma$.

We now rewrite $\PD{}{t}\paren{\PD{w}{\mb{n}}}$ in 
\eqref{identity} in an analogous fashion.
For this, we first introduce the signed distance function 
$\psi(\mb{x},t)$ in a neighborhood
of $\Gamma$:
\begin{equation}
\psi(\mb{x},t)=
\begin{cases}
\text{dist}(\mb{x},\Gamma_t) &\text{ if } \mb{x}\in \Omega_e,\\
0 &\text{ if } \mb{x}\in \Gamma_t,\\
-\text{dist}(\mb{x},\Gamma_t) &\text{ if } \mb{x}\in \Omega_i,
\end{cases}
\end{equation}
where $\text{dist}(\mb{x},\Gamma_t)$ is the distance between 
$\mb{x}$ and $\Gamma_t$.
Clearly, $\nabla\psi$
evaluated at any point on $\Gamma$ gives the outward unit 
normal vector $\mb{n}$. 
Introduce the following vector field 
$\mb{v}$ defined in a neighborhood of $\Gamma$ where $\psi$
is smooth:
\begin{equation}
\mb{v}=v_\Gamma\mb{n} \text{ on } \Gamma,\; (\nabla \mb{v})\nabla \psi=0.
\label{defv}
\end{equation}
The second condition above just says that $\mb{v}$ is constant 
along lines perpendicular to the level sets.  
It is well known that the signed distance function satisfies the 
following transport equation in a neighborhood of $\Gamma$:
\begin{equation}
D_v\psi\equiv \PD{\psi}{t}+\mb{v}\cdot\nabla \psi=0.\label{vtransport}
\end{equation}
Note that the above convective derivative evaluated on $\Gamma$
is equal to $D^\mb{n}_t$ defined in \eqref{Dnt}.

For any point on $\Gamma$:
\begin{equation}
\PD{}{\mb{n}}\paren{\PD{w}{t}}=\nabla \psi\cdot \nabla w_t
\end{equation}
where the subscript $t$ indicates the partial derivative 
with respect to $t$. 
We now rewrite this expression as follows:
\begin{equation}
\begin{split}
\nabla \psi\cdot \nabla w_t
&=D_v(\nabla \psi\cdot \nabla w)
-\nabla \psi_t\cdot \nabla w
-\mb{v}\cdot \nabla (\nabla \psi\cdot \nabla w)\\
&=D_v(\nabla \psi\cdot \nabla w)
+\nabla (\mb{v}\cdot \nabla \psi)\cdot \nabla w
-\mb{v}\cdot \nabla (\nabla \psi\cdot \nabla w)
\end{split}\label{psiphit}
\end{equation}
where we used \eqref{vtransport} in the last equality.
Now, consider the second term in the last line:
\begin{equation}
\begin{split}
\nabla (\mb{v}\cdot \nabla \psi)\cdot \nabla w
=&\nabla_\Gamma(\mb{v}\cdot \nabla \psi)\cdot \nabla_\Gamma w
+(\nabla \psi\cdot \nabla (\mb{v}\cdot \nabla \psi))
(\nabla \psi\cdot \nabla w)\\
=&\nabla_\Gamma(\mb{v}\cdot \nabla \psi)\cdot \nabla_\Gamma w
+(\nabla \psi\cdot((\nabla \mb{v}) \nabla \psi))
(\nabla \psi\cdot \nabla w)\\
&+(\mb{v}\cdot((\nabla^2 \psi) \nabla \psi))
(\nabla \psi\cdot \nabla w)
\end{split}\label{vpsiphi1}
\end{equation}
where $\nabla_\Gamma$ is the surface gradient on $\Gamma$.
Note that $\nabla^2\psi$ is {\em not} the Laplacian but the 
matrix of second derivatives of $\psi$.
The second to last term in \eqref{vpsiphi1} is $0$ by \eqref{defv}.
The last term is also $0$, since:
\begin{equation}
(\nabla^2 \psi) \nabla \psi
=\frac{1}{2}\nabla(\abs{\nabla \psi}^2)=0\label{d2psi0}
\end{equation}
where we used $\abs{\nabla\psi}^2=1$.
Thus \eqref{vpsiphi1} reduces to
\begin{equation}
\nabla (\mb{v}\cdot \nabla \psi)\cdot \nabla w
=\nabla_\Gamma(\mb{v}\cdot \nabla \psi)\cdot \nabla_\Gamma w
\label{vpsiphi11}.
\end{equation}
Let us look at the final term in \eqref{psiphit}. 
\begin{equation}
\mb{v}\cdot \nabla(\nabla \psi\cdot \nabla w)
=(\mb{v}\cdot \nabla \psi)(\nabla \psi\cdot(\nabla \psi \cdot \nabla w))
=(\mb{v}\cdot \nabla \psi)(\nabla \psi\cdot((\nabla^2w)\nabla \psi))
\label{vpsiphi2}
\end{equation}
where we used \eqref{d2psi0} in the last equality.
Combining \eqref{vpsiphi11}, \eqref{vpsiphi2} and \eqref{psiphit}, we have:
\begin{equation}
\nabla \psi\cdot \nabla w_t
=D_v(\nabla \psi\cdot \nabla w)
+\nabla_\Gamma(\mb{v}\cdot \nabla \psi)\cdot \nabla_\Gamma w
-(\mb{v}\cdot \nabla \psi)(\nabla \psi\cdot((\nabla^2w)\nabla \psi))
\end{equation}
or equivalently:
\begin{equation}
\PD{}{\mb{n}}\paren{\PD{w}{t}}
=D_t^\mb{n}\paren{\PD{w}{\mb{n}}}
+\nabla_\Gamma v_\Gamma \cdot \nabla_\Gamma w
-v_\Gamma\paren{\PDD{2}{w}{\mb{n}}}
\label{dtdphidn}
\end{equation}
where we used \eqref{defv}, $\mb{n}=\nabla \psi$ on $\Gamma$ and 
the equality of $D_t^\mb{n}$ and $D_v$ on $\Gamma$.

Now, consider the integral:
\begin{equation}
\begin{split}
&\int_{\Gamma_i}\paren{w\PD{}{\mb{n}}\paren{\PD{w}{t}}
+(w\Delta w)v_\Gamma}dm_\Gamma\\
=&\int_{\Gamma_i}\paren{w D_t^\mb{n}\paren{\PD{w}{\mb{n}}}
+w\nabla_\Gamma v_\Gamma\cdot \nabla_\Gamma w
+w\paren{\Delta_\Gamma w+\kappa_\Gamma \PD{w}{\mb{n}}}v_\Gamma} dm_\Gamma\\
=&\int_{\Gamma_i}\paren{w D_t^\mb{n}\paren{\PD{w}{\mb{n}}}
+\paren{\kappa_\Gamma w\PD{w}{\mb{n}}
-\abs{\nabla_\Gamma w}^2}v_\Gamma}dm_\Gamma
\end{split}
\end{equation}
where we used \eqref{lapphi} and \eqref{dtdphidn} in the first 
equality and integrated by parts along $\Gamma$ in the second 
equality. Note that there are no boundary terms since $\Gamma$
is a closed compact surface. This proves \eqref{identity}.
\end{proof}

We are now ready to prove Theorem \ref{main}.

\begin{proof}[Proof of Theorem \ref{main}]
First, multiply \eqref{ckeq} with $\mu_k$ in \eqref{muk} and integrate 
over $\Omega_i$ and sum in $k$:
\begin{equation}
\sum_{k=1}^N\int_{\Omega_i}\mu_k\paren{\PD{c_k}{t}+\nabla \cdot(\mb{u}c_k)}d\mb{x}
=\sum_{k=1}^N\int_{\Omega_i}\mu_k \nabla \cdot \paren{c_k\frac{D_k}{k_BT}\nabla \mu_k}d\mb{x}.\label{mkcke}
\end{equation}
The summand in the right hand side becomes:
\begin{equation}
\begin{split}
\int_{\Omega_i}\mu_k \nabla \cdot \paren{c_k\frac{D_k}{k_BT}\nabla \mu_k}d\mb{x}
&=\int_{\Gamma_i}\paren{\mu_k {c_k\frac{D_k}{k_BT}\nabla \mu_k}\cdot \mb{n}}dm_\Gamma\\
&-\int_{\Omega_i}\paren{c_k\frac{D_k}{k_BT}\abs{\nabla \mu_k}^2}d\mb{x}
\end{split}
\end{equation}
where $\mb{n}$ is the outward normal on $\Gamma$. Consider the 
left hand side of \eqref{mkcke}.
\begin{equation}
\begin{split}
&\sum_{k=1}^N\mu_k\PD{c_k}{t}=\sum_{k=1}^N\paren{\sigma_k\PD{c_k}{t}+qz_k\phi\PD{c_k}{t}}\\
=&\sum_{k=1}^N\PD{\omega}{c_k}\PD{c_k}{t}+\phi\PD{}{t}\paren{\sum_{k=1}^N qz_kc_k}
=\PD{\omega}{t}-\phi\PD{}{t}\paren{\nabla \cdot\paren{\epsilon\nabla \phi}}
\end{split}\label{mkcksume}
\end{equation}
We used \eqref{muk} in the first equality and \eqref{poisson} in the last equality.
Integrate final expression in \eqref{mkcksume} over $\Omega_i$.
\begin{equation}
\begin{split}
&\int_{\Omega_i}\paren{\PD{\omega}{t}-\phi\nabla\cdot
\paren{\epsilon\nabla\paren{\PD{\phi}{t}}}} d\mb{x}\\
=&\int_{\Gamma_i}\paren{-\phi\PD{}{\mb{n}}\paren{\epsilon\PD{\phi}{t}}}dm_\Gamma
+\int_{\Omega_i} \PD{}{t}\paren{\omega+\frac{\epsilon}{2}\abs{\nabla \phi}^2}d\mb{x}.
\end{split}
\end{equation}
For the second term in the left hand side of \eqref{mkcke}, we have, 
similarly to \eqref{mkcksume}:
\begin{equation}
\sum_{k=1}^N\mu_k\mb{u}\cdot \nabla c_k=\nabla \cdot (\mb{u}\omega) 
+\phi\cdot \nabla\paren{\mb{u} \sum_{k=1}^Nqz_kc_k}.
\end{equation}
Integrate the above expression over $\Omega_i$:
\begin{equation}
\begin{split}
&\int_{\Omega_i}\paren{\nabla \cdot (\mb{u}\omega) 
+\phi\cdot \nabla\paren{\mb{u}\sum_{k=1}^Nqz_kc_k}}d\mb{x}\\
=&\int_{\Gamma_i}\paren{\omega+\phi\sum_{k=1}^Nqz_kc_k}
\mb{u}\cdot\mb{n}dm_\Gamma
-\int_{\Omega_i}\paren{\sum_{k=1}^Nqz_kc_k}\mb{u}\cdot \nabla \phi d\mb{x}
\end{split}
\end{equation}
Collecting the above calculations, 
we have rewritten identity \eqref{mkcke} as:
\begin{equation}
\begin{split}
&\int_{\Omega_i} 
\PD{}{t}\paren{\omega+\frac{\epsilon}{2}\abs{\nabla \phi}^2}d\mb{x}
+\int_{\Gamma_i} \paren{-\phi\PD{}{\mb{n}}\paren{\epsilon\PD{\phi}{t}}}
dm_\Gamma\\
=&-\int_{\Gamma_i}
\paren{\paren{\omega-\sum_{k=1}^N c_k\sigma_k}\mb{u}\cdot\mb{n}
+\sum_{k=1}^N \mu_k\paren{c_k\paren{\mb{u}
-\frac{D_k}{k_BT}\nabla \mu_k}\cdot \mb{n}}}
dm_\Gamma\\
+&\int_{\Omega_i}\paren{-\sum_{k=1}^Nc_k\frac{D_k}{k_BT}\abs{\nabla \mu_k}^2
+\paren{\sum_{k=1}^Nqz_kc_k}\mb{u}\cdot \nabla \phi}d\mb{x}
\end{split}
\end{equation}
where we used $qz_k\phi=\mu_k-\sigma_k$ (Eq. \eqref{muk}) in the 
boundary integral after the equality.
Performing a similar calculation on $\Omega_e$, and adding this to 
the above, we find:
\begin{equation}
\begin{split}
&\D{}{t}\int_{\Omega_i\cup \Omega_e}
\paren{\omega+\frac{\epsilon}{2}\abs{\nabla \phi}^2}d\mb{x}
-\int_{\Gamma}\jump{\phi\PD{}{\mb{n}}\paren{\epsilon\PD{\phi}{t}}}dm_\Gamma\\
-&\int_{\Gamma}\paren{
\jump{\omega+\frac{\epsilon}{2}\abs{\nabla \phi}^2}\PD{\mb{X}}{t}\cdot\mb{n}
}dm_\Gamma\\
=&-\int_{\Gamma}\paren{\jump{\pi_w}\mb{u}\cdot\mb{n}
+\sum_{k=1}^N\paren{[\mu_k](j_k+a_k)+[c_k\mu_k]\PD{\mb{X}}{t}\cdot\mb{n}}}
dm_\Gamma\\
+&\int_{\Omega_i}\paren{-\sum_{k=1}^Nc_k\frac{D_k}{k_BT}\abs{\nabla \mu_k}^2
+\paren{\sum_{k=1}^Nqz_kc_k}\mb{u}\cdot \nabla \phi}d\mb{x}
\end{split}
\end{equation}
where we used \eqref{muw} and \eqref{ckbc} 
to rewrite the boundary integral after 
the equality. Note that the second boundary integral before the equality 
comes from the fact that the boundary $\Gamma$ is moving.
Rearranging terms and using \eqref{cont}, we have:
\begin{equation}
\begin{split}
&\D{}{t}\int_{\Omega_i\cup \Omega_e}
\paren{\omega+\frac{\epsilon}{2}\abs{\nabla \phi}^2}d\mb{x}\\
=&\int_{\Gamma}\paren{
\jump{\phi\PD{}{\mb{n}}\paren{\epsilon\PD{\phi}{t}}}
+\jump{\frac{\epsilon}{2}\abs{\nabla \phi}^2
+\phi\nabla\cdot(\epsilon\nabla \phi)
}\PD{\mb{X}}{t}\cdot \mb{n}}dm_\Gamma\\
-&\int_{\Gamma}\paren{[\pi_w]j_w
+\sum_{k=1}^N[\mu_k](j_k+a_k)}dm_\Gamma\\
+&\int_{\Omega_i\cup\Omega_e}\paren{-\sum_{k=1}^Nc_k\frac{D_k}{k_BT}\abs{\nabla \mu_k}^2
+\paren{\sum_{k=1}^Nqz_kc_k}\mb{u}\cdot \nabla \phi}d\mb{x}.
\end{split}\label{main1}
\end{equation}
We used $\mu_k-\sigma_k=qz_k\phi$ and used \eqref{poisson} to rewrite 
the first boundary integral after the equality.
Note that:
\begin{equation}
\begin{split}
&\int_{\Gamma}
\jump{\phi\PD{}{\mb{n}}\paren{\epsilon\PD{\phi}{t}}
+\phi\nabla\cdot(\epsilon\nabla \phi)
\PD{\mb{X}}{t}\cdot\mb{n}}dm_\Gamma\\
=&\int_{\Gamma}
\jump{\phi D_t^\mb{n}\paren{\epsilon \PD{\phi}{\mb{n}}}
+\paren{\kappa_\Gamma \phi\epsilon\PD{\phi}{\mb{n}} 
-\epsilon \abs{\nabla_\Gamma \phi}^2}
\PD{\mb{X}}{t}\cdot\mb{n}}dm_\Gamma
\end{split}
\end{equation}
where we used Lemma \ref{calcidentity} with $w=\phi$ and 
$v_\Gamma=\PD{\mb{X}}{t}\cdot \mb{n}$ in \eqref{identity}.
Using this and the definition of $\nabla_\Gamma$, we may 
rewrite the first boundary integral in \eqref{main1} as:
\begin{equation}
\begin{split}
&\int_{\Gamma}\paren{
\jump{\phi\PD{}{\mb{n}}\paren{\epsilon\PD{\phi}{t}}}
+\jump{\frac{\epsilon}{2}\abs{\nabla \phi}^2
+\phi\nabla\cdot(\epsilon\nabla \phi)
}\PD{\mb{X}}{t}\cdot\mb{n}}dm_\Gamma\\
=&-\int_{\Gamma}[\phi]D_t^\mb{n}(C_m[\phi])dm_\Gamma\\
&+\int_{\Gamma}\paren{-\kappa_\Gamma C_m[\phi]^2
+\jump{\frac{\epsilon}{2}\paren{\abs{\PD{\mb{\phi}}{\mb{n}}}^2
-\abs{\nabla_\Gamma \phi}^2}}}\PD{\mb{X}}{t}\cdot\mb{n}
dm_\Gamma\label{phibterms}
\end{split}
\end{equation}
where we used \eqref{cap}.

We now turn to equation \eqref{stokesE}. Multiply this by $\mb{u}$
and integrate over $\Omega_i$:
\begin{equation}
\begin{split}
&\int_{\Omega_i} \mb{u}\cdot(\nu \Delta \mb{u} -\nabla p)d\mb{x}
-\int_{\Omega_i}\paren{\sum_{k=1}^Nqz_kc_k}\mb{u}\cdot \nabla \phi d\mb{x}\\
=&\int_{\Gamma_i}\paren{\Sigma(\mb{u},p)\mb{n}}\cdot\mb{u}dm_\Gamma-
\int_{\Omega_i}\mb{\nu}\abs{\nabla \mb{u}}^2d\mb{x}\\
&-\int_{\Omega_i}
\paren{\sum_{k=1}^Nqz_kc_k}\mb{u}\cdot \nabla \phi d\mb{x}=0
\end{split}
\end{equation}
Performing a similar calculation on $\Omega_e$ and by summation, we have:
\begin{equation}\label{sigupgamma}
\int_{\Gamma}\jump{\paren{\Sigma(\mb{u},p)\mb{n}}\cdot\mb{u}}dm_\Gamma-
\int_{\Omega_i\cup \Omega_e}\mb{\nu}\abs{\nabla \mb{u}}^2d\mb{x}
=\int_{\Omega_i\cup \Omega_e}
\paren{\sum_{k=1}^Nqz_kc_k}\mb{u}\cdot \nabla \phi d\mb{x}
\end{equation}

Let us first assume \eqref{forcebalanceE} holds. 
First write $\Sigma_e(\phi)\mb{n}$ in the following form:
\begin{equation}
\begin{split}
\Sigma_e(\phi)\mb{n}&=\epsilon\paren{\PD{\phi}{\mb{n}}\nabla \phi-\frac{1}{2}\abs{\nabla \phi}^2\mb{n}}\\
&=\epsilon\paren{
\frac{1}{2}\paren{\abs{\PD{\phi}{\mb{n}}}^2-\abs{\nabla_\Gamma\phi}^2}\mb{n}
+\PD{\phi}{\mb{n}}\nabla_\Gamma\phi}
\end{split}
\end{equation}
We may now write \eqref{forcebalanceE} as:
\begin{equation}
[\Sigma_m(\mb{u},p)\mb{n}]=\mb{F}_\text{elas}+\mb{F}_\text{cap}-
\jump{\frac{\epsilon}{2}\paren{\abs{\PD{\phi}{\mb{n}}}^2
-\abs{\nabla_\Gamma\phi}^2}}\mb{n}
+C_m[\phi]\nabla_\Gamma [\phi]\label{bup}
\end{equation}
where we used \eqref{cap} in the last term.

Using \eqref{phibterms}, \eqref{sigupgamma}
and \eqref{bup} in \eqref{main1} we have:
\begin{equation}
\begin{split}
&\D{}{t}\int_{\Omega_i\cup \Omega_e}
\paren{\omega+\frac{\epsilon}{2}\abs{\nabla \phi}^2}d\mb{x}\\
=&\int_{\Gamma}\paren{-[\phi]D_t^\mb{n}(C_m[\phi])
+\paren{-\kappa_\Gamma C_m[\phi]^2\mb{n}+C_m[\phi]\nabla_\Gamma[\phi]+\mb{F}_\text{cap}}
\cdot \PD{\mb{X}}{t}}dm_\Gamma\\
+&\int_\Gamma \mb{F}_\text{elas}\cdot\PD{\mb{X}}{t}dm_\Gamma
-\int_{\Gamma}\paren{[\widehat{\psi_w}]j_w
+\sum_{k=1}^N[\mu_k](j_k+a_k)}dm_\Gamma\\
-&\int_{\Omega_i\cup\Omega_e}\paren{\sum_{k=1}^Nc_k\frac{D_k}{k_BT}\abs{\nabla \mu_k}^2
+\mb{\nu}\abs{\nabla \mb{u}}^2}d\mb{x}.
\end{split}\label{main2}
\end{equation} 
Comparing \eqref{main2}, \eqref{FEE} and using \eqref{dtelas}, 
we see that the proof 
of \eqref{FEE} rests on the evaluation of the first 
boundary integral after the equality in \eqref{main2}.
\begin{equation}
\begin{split}
&\int_{\Gamma}[\phi]D_t^\mb{n}(C_m[\phi])dm_\Gamma
=\int_{\Gamma}
\paren{D_t^\mb{n}\paren{\frac{1}{2}C_m[\phi]^2}
+\frac{1}{2}(D_t^\mb{n}C_m)[\phi]^2}dm_\Gamma\\
=&\D{}{t}\int_\Gamma \frac{1}{2}C_m[\phi]^2dm_\Gamma
-\int_\Gamma\paren{
\paren{\frac{1}{2}C_m[\phi]^2}\kappa_\Gamma \PD{\mb{X}}{t}\cdot\mb{n}
-\frac{1}{2}(D_t^\mb{n}C_m)[\phi]^2}dm_\Gamma\label{phiDtnphi}
\end{split}
\end{equation}
where we used \eqref{Dtw} with $w=[\phi]$, 
$v_\Gamma=\PD{\mb{X}}{t}\cdot \mb{n}$ in the second equality.
We also have:
\begin{equation}
\int_{\Gamma}(C_m[\phi]\nabla_\Gamma [\phi])dm_\Gamma
=\int_{\Gamma}\paren{\nabla_\Gamma \paren{\frac{1}{2}C_m[\phi]^2}-
\frac{1}{2}(\nabla_\Gamma C_m)[\phi]^2}dm_\Gamma\label{phiCmphi}
\end{equation}
Using \eqref{phiDtnphi} and \eqref{phiCmphi}, we have:
\begin{equation}
\begin{split}
&\int_{\Gamma}\paren{-[\phi]D_t^\mb{n}(C_m[\phi])
+\paren{-\kappa_\Gamma C_m[\phi]^2\mb{n}+C_m[\phi]\nabla_\Gamma[\phi]}\cdot \PD{\mb{X}}{t}}dm_\Gamma\\
=&-\D{}{t}\int_\Gamma \frac{1}{2}C_m[\phi]^2dm_\Gamma
-\int_\Gamma \frac{1}{2}
\paren{D_t^\mb{n}C_m+(\nabla_\Gamma C_m)\cdot\PD{\mb{X}}{t}}[\phi]^2 dm_\Gamma
\\
&+\int_\Gamma\paren{
-\paren{\frac{1}{2}C_m[\phi]^2}\kappa_\Gamma\mb{n}+
\nabla_\Gamma \paren{\frac{1}{2}C_m[\phi]^2}}
\cdot \PD{\mb{X}}{t}dm_\Gamma\\
\end{split}\label{main3}
\end{equation}
Consider the second boundary integral after the equality. 
\begin{equation}
\begin{split}
&\int_\Gamma \frac{1}{2}
\paren{D_t^\mb{n}C_m+(\nabla_\Gamma C_m)\cdot\PD{\mb{X}}{t}}[\phi]^2 dm_\Gamma\\
=&\int_\Gamma \paren{\frac{1}{2}\PD{C_m}{t}[\phi]^2}dm_\Gamma
=\int_\Gamma \paren{\frac{1}{2}\PD{C_m}{Q}\PD{Q}{t}[\phi]^2}dm_\Gamma\\
=&\int_\Gamma \paren{\frac{1}{2}Q\PD{C_m}{Q}[\phi]^2\kappa_\Gamma\mb{n} 
-\nabla_\Gamma\paren{\frac{1}{2}Q\PD{C_m}{Q}[\phi]^2}}\cdot\PD{\mb{X}}{t}
dm_\Gamma
\end{split}\label{Cmtterm}
\end{equation}
where we used \eqref{wt} with $w=C_m$ in the first equality, 
and \eqref{Qt} with $w=\frac{1}{2}Q\PD{C_m}{Q}[\phi]^2$ in the last equality.
From \eqref{Cmtterm}, \eqref{main3}, \eqref{main2} and expression 
\eqref{FcapCmQ} of $\mb{F}_\text{cap}$, we obtain the desired result.

If $Q\equiv 1$ and \eqref{forcebalancep} holds, we may argue as follows. 
Equation \eqref{main2} remains valid with $\mb{F}_\text{cap}$ replaced 
by $\mb{F}_\text{p}$. Verification of \eqref{FEE} rests on the 
evaluation of the first boundary integral in \eqref{main2}.
Proceeding as in the above, we have:
\begin{equation}
\begin{split}
&\int_{\Gamma}\paren{-[\phi]D_t^\mb{n}(C_m[\phi])
+\paren{-\kappa_\Gamma C_m[\phi]^2\mb{n}+C_m[\phi]\nabla_\Gamma[\phi]+\mb{F}_p}\cdot \PD{\mb{X}}{t}}dm_\Gamma\\
=&-\D{}{t}\int_\Gamma \frac{1}{2}C_m[\phi]^2dm_\Gamma
-\int_\Gamma \paren{\frac{1}{2}\PD{C_m}{Q}\PD{Q}{t}[\phi]^2}dm_\Gamma\\
&+\int_\Gamma\paren{
\paren{\lambda-\frac{1}{2}C_m[\phi]^2}\kappa_\Gamma\mb{n}-
\nabla_\Gamma \paren{\lambda-\frac{1}{2}C_m[\phi]^2}}
\cdot \PD{\mb{X}}{t}dm_\Gamma
\end{split}
\end{equation}
where we used the \eqref{Fp}.
Since $\PD{Q}{t}=0$, the second boundary integral after the equality 
is $0$. Using \eqref{Qt} with $w=\lambda-\frac{1}{2}C_m[\phi]^2$ 
and $\PD{Q}{t}=0$, we see that the 
last boundary integral is also $0$.

In the absence of active currents $a_k$, the free energy is decreasing
since $j_k$ and $j_w$ satisfying conditions 
\eqref{jkcond} and \eqref{jwcond} respectively.  
\end{proof}

\section{Limiting Systems}\label{simple}

We now discuss some limiting cases of the model we introduced in the 
previous section. 
For this purpose, we shall first make the equations dimensionless.
In what follows, the primed symbols denote dimensionless variables.
We introduce the following non-dimensionalization of space and time.
\begin{equation}
\mb{x}=L\mb{x}',\; \mb{X}=L\mb{X}',\; t=T_Dt',\; T_D=\frac{L^2}{D_0},\; D_k=D_0D_k',
\end{equation}
where $L$ is the characteristic length scale (for example the 
size of the domain $\Omega_i$) and $D_0$ is the characteristic 
diffusion coefficients of ions. We thus measure time with respect 
to the diffusive time scale of ions. For concentrations 
and the electrostatic potential, we let:
\begin{equation}
c_k=c_0c_k',\; \phi=\frac{k_BT}{q}\phi'.
\end{equation}
For pressure and the membrane elastic force, we let:
\begin{equation}
p=c_0k_BTp',\; \mb{F}_\text{elas}=c_0k_BT\mb{F}_\text{elas}'.
\end{equation}
For the characteristic fluid and membrane velocity, we turn to 
relation \eqref{cont}. Let $\zeta$ be the characteristic hydraulic 
permeability of the membrane, which we may take as follows: 
\begin{equation}
\zeta=\at{\PD{j_w}{[\widehat{\psi_w}]}}
{[\widehat{\psi_w}]=0}.
\end{equation}
Then, $\zeta c_0k_BT$ is the characteristic 
velocity generated by an osmotic gradient across the membrane.
We thus let:
\begin{equation}
\mb{u}=\zeta c_0k_BT \mb{u}'.
\end{equation}
With the above dimensionless variables, we may rewrite our system as follows.
For simplicity, we shall adopt expression \eqref{ent} 
as our definition of the entropic part of the free energy $\omega$, 
so that:
\begin{equation}
\mu_k'=z_k\phi'+\ln c_k'.
\end{equation}
In $\Omega_i$ and $\Omega_e$, we have:
\begin{subequations}\label{dlessfull}
\begin{align}
\PD{c_k'}{t'}+\text{Pe}\nabla'\cdot (\mb{u}'c_k')&=-\nabla'\cdot \mb{f}_k',\;
\mb{f}_k'=-D_k'(\nabla'c_k'+z_kc_k'\nabla'\phi'),\label{dlessck}\\
-\nabla'\cdot(\beta^2\nabla'\phi')&=\sum_{k=1}^N z_kc_k',\label{dlesspoisson}\\
\gamma\Delta' \mb{u}'&=\nabla'p'+\paren{\sum_{k=1}^N z_kc_k'}\nabla \phi', \; 
\nabla'\cdot \mb{u}'=0,\label{dlessstokes}
\end{align}
where $\nabla', \nabla'\cdot$ and $\Delta'$ are the gradient, 
divergence and Laplace operators in the $\mb{x}'$ coordinate
and the dimensionless constants are:
\begin{equation}
\text{Pe}=\frac{\zeta c_0k_BT}{D/L},\; \beta=\frac{r_d}{L},\; 
r_d=\sqrt{\frac{\epsilon k_BT}{q^2c_0}},\; \gamma=\frac{\nu \zeta}{L}. 
\label{Pebeta}
\end{equation}
In the above, $\text{Pe}$ is the P\'eclet number which, in this case, measures 
the ratio between the fluid velocity induced by osmotic gradients and 
the characteristic diffusive velocity. The constant $\beta$ measures 
the ratio between $r_d$, the Debye length and $L$. The constant $\gamma$
is the ratio between the viscosity of water and the hydraulic resistance 
of the membrane. 
The boundary conditions at the membrane interface $\Gamma$ become:
\begin{align}
\at{\paren{\paren{c_k'\paren{\text{Pe}\mb{u}'-\PD{\mb{X}'}{t'}}+\mb{f}'_k}
\cdot \mb{n}}}{\Gamma_{i,e}}&=\alpha (j_k'+a_k'),\label{dlessckbc}\\
-\at{\paren{\beta\nabla'{\phi'}\cdot \mb{n}}}{\Gamma_{i,e}}&=
\theta C_m'\jump{\phi'},\label{dlessbcpoisson}\\
\mb{u}'-\frac{1}{\text{Pe}}\PD{\mb{X}'}{t'}&=j_w'\mb{n},\label{dlesscont}\\
\jump{\paren{\Sigma'_m(\mb{u}',p')+\beta^2\Sigma_e'(\phi')}\mb{n}}
&=\mb{F}'_\text{elas}+\beta\theta\mb{F}'_\text{cap}.\label{dlessforce}
\end{align}
\end{subequations}
In equation \eqref{dlessckbc}, $\alpha$ is a dimensionless constant 
given by the ratio of the characteristic membrane permeability $p_m$
and diffusion in the bulk:
\begin{equation}\label{alphapm}
\alpha=\frac{p_m}{D/L}, \; p_m=\sum_{k=1}^N\frac{k_BT}{c_0}\at{\PD{j_k}{\jump{\mu_k}}}
{\jump{\mu_k}=0}.
\end{equation}
The currents $j_k$ and $a_k$
are scaled so that $j_k=p_mc_0j_k'$ and $a_k=p_mc_0a_k'$.
In \eqref{dlessbcpoisson}:
\begin{equation}
C_m=C_m^0C_m', \; \theta=\frac{C_m^0 k_BT/q}{qc_0r_d},
\end{equation}
where $C_m^0$ is the characteristic magnitude of the membrane 
capacitance per unit area (see \eqref{Cm0} or \eqref{Cmscaling}).
The dimensionless constant $\theta$ is the ratio between the 
membrane charge and the total amount of charge in a layer of 
thickness on the order of the Debye length. In \eqref{dlesscont}, 
$j_w=\zeta c_0k_BT j_w'$. The variables in \eqref{dlessforce},
are defined by:
\begin{align}
\Sigma'_m(\mb{u}',p')&=\gamma(\nabla' \mb{u}'+(\nabla' \mb{u}')^T)-p'I,\\
\Sigma'_e(\phi')&=\nabla' \phi'\otimes \nabla \phi'
-\frac{1}{2}\abs{\nabla'\phi'}^2I,\\
\mb{F}_\text{cap}&=
\tau_\text{cap}'\kappa_\Gamma'\mb{n}-\nabla_\Gamma'\tau_\text{cap}', 
\; \tau_\text{cap}'
=\frac{1}{2}\paren{C_m'+Q\PD{C_m'}{Q}}\jump{\phi'}^2
\end{align}
where $\kappa_\Gamma'(=\kappa_\Gamma L)$ 
is the sum of the two principal curvatures of $\Gamma$ measured 
in the $\mb{x}'$ spatial variable and $\nabla_\Gamma'$ is the surface 
gradient operator with respect to $\mb{x}'$. 
Equations \eqref{dlessck}-\eqref{dlessstokes} and the boundary conditions 
\eqref{dlessckbc}-\eqref{dlessforce} constitute our dimensionless 
system. In the rest of this section we shall 
drop the $\prime$ in the variables with the understanding that 
all variables, unless otherwise stated, are dimensionless. 

The dimensionless system above possesses five dimensionless constants 
$\alpha, \beta, \gamma, \theta$ and $\text{Pe}$. We consider two 
limiting cases of the above system. First of all, consider the case 
when $\text{Pe}\ll 1$. Assuming all primed quantities are $\mathcal{O}(1)$
with respect to $\text{Pe}$, we see, from \eqref{dlesscont} that 
\begin{equation}
\PD{\mb{X}}{t}=\mathcal{O}(\text{Pe}).
\end{equation}
Therefore, the membrane does not move to leading order. 
If we collect all leading order terms, we see that the equations 
\eqref{dlessck} and \eqref{dlesspoisson} decouple 
from \eqref{dlessstokes}. We thus obtain the following Poisson-Nernst-Planck 
system with interface boundary conditions:
\begin{subequations}\label{Pe0}
\begin{align}
\PD{c_k}{t}&=-\nabla\cdot \mb{f}_k \text{ in } \Omega_{i,e}\\
-\nabla\cdot(\beta^2\nabla\phi)&=\sum_{k=1}^N z_kc_k 
\text{ in } \Omega_{i,e}\\
\at{(\mb{f}_k\cdot \mb{n})}{\Gamma_{i,e}}&=\alpha (j_k+a_k),\\
-\at{\paren{\beta\nabla{\phi}\cdot \mb{n}}}{\Gamma_{i,e}}&=
\theta C_m\jump{\phi},
\end{align}
\end{subequations}
where the membrane $\Gamma$ is fixed in time. This model was introduced in 
\cite{leonetti_biomembrane_1998} 
(see also \cite{schaff1997general,choi1999electrodiffusion} 
for related models).

For single cell systems, the P\'eclet number is about 
$\text{Pe}\approx 10^{-1}$ to $10^{-4}$. 
The above may thus be a good approximation to 
the full system in the $T_D$ time scale. In the context of 
multicellular systems, however, $L$ may be large and 
$\text{Pe}$ can reach unity, 
as can be seen from expression \eqref{Pebeta} of $\text{Pe}$.
It should be pointed out that there are situations in which the representative 
fluid velocity is not dictated by the osmotic pressure, in which 
case one should adopt a different definition for the P\'eclet number. 
For example, if we are interested in blood cells in a flow 
environment, the ambient hemodynamic flow velocity should be 
taken as the representative velocity. 

We note that \eqref{Pe0} also satisfies a free energy equality.
\begin{proposition}
Suppose $c_k$ and $\phi$ are smooth functions that satisfy \eqref{Pe0}.
Then, the following equality holds:
\begin{equation}
\D{}{t}\paren{G_S+E_\text{elec}}=-F_c-J_a.\label{PeFE}
\end{equation}
In the above, $G_S, E_\text{elec}, F_c$ and $J_a$ 
are dimensionless versions of the corresponding quantities 
in \eqref{FE}, \eqref{Fwc} and \eqref{FEE}.  
\end{proposition} 
\begin{proof}
This follows from a simple calculation.
\end{proof}
In this sense, system \eqref{Pe0} may be seen as the system 
associated with the energy law \eqref{PeFE} where the mechanical 
energy and dissipation in \eqref{FEE} is discarded.

We next consider the limit when $\beta\ll 1$. This limit is motivated 
by the fact that the Debye length $r_d$ is approximately $1$nm
in typical physiological systems, 
far smaller than the typical length scale of interest.
By formally letting $\beta \to 0$ in \eqref{dlessfull}
we obtain the following system of equations:
\begin{subequations}\label{beta0}
\begin{align}
\PD{c_k}{t}+\text{Pe}\mb{u}\cdot\nabla c_k&=-\nabla\cdot \mb{f}_k
\text{ in } \Omega_{i,e} \label{dlessckbeta0}\\
\quad \sum_{k=1}^N z_kc_k&=0 \text{ in } \Omega_{i,e},\label{dlessEN}\\
\gamma\Delta \mb{u}-\nabla p&=0, 
\quad \nabla\cdot \mb{u}=0 \text{ in } \Omega_{i,e},\label{beta0stokes}\\
\alpha (j_k+a_k)&=\at{\paren{\paren{c_k\paren{\text{Pe}\mb{u}-\PD{\mb{X}}{t}}+\mb{f}_k}
\cdot \mb{n}}}{\Gamma_{i,e}},\label{dlessbeta0ckbc}\\
\mb{u}-\frac{1}{\text{Pe}}\PD{\mb{X}}{t}&=j_w\mb{n},\quad
\jump{\Sigma_m(\mb{u},p)\mb{n}}=\mb{F}_\text{elas} \text{ on } \Gamma.
\label{beta0stokesbc}
\end{align}
\end{subequations}
We have discarded all terms in \eqref{dlessfull} that 
involve $\beta$ and have eliminated the boundary condition 
\eqref{dlessbcpoisson}. The most important feature 
of the above system is that we have, in place of the Poisson 
equation \eqref{dlesspoisson}, the electroneutrality condition
\eqref{dlessEN}.
The electrostatic potential $\phi$ thus evolves so that 
the electroneutrality constraint \eqref{dlessEN} is satisfied 
at each time instant. 
Although $\phi$ is thus determined only implicitly through 
the electroneutrality condition, it is possible to obtain 
a PDE satisfied by $\phi$ by taking the derivative of \eqref{dlessEN}
with respect to $t$ and using \eqref{dlessckbeta0}:
\begin{equation}\label{phieqbeta0}
\begin{split}
0&=\nabla \cdot(a\nabla \phi+\mb{b})\\
a&=\sum_{k=1}^N z_k^2D_kc_k, \; 
\mb{b}=\sum_{k=1}^N z_kD_k\nabla c_k.
\end{split}
\end{equation}
We point out that the electroneutrality 
condition does {\em not} imply that $\Delta \phi=0$ as 
may be erroneously inferred from \eqref{dlesspoisson}. 
In fact, as $\beta\to 0$, $\Delta \phi$ may remain order 
$1$ with respect to $\beta$ while the right hand side of 
\eqref{dlesspoisson} will go to $0$ like $\beta^2$.
This is a common fallacy in applications of the electroneutral limit. 
The boundary condition for this elliptic equation can be obtained 
by taking the sum in $k$ in boundary condition \eqref{dlessbeta0ckbc}:
\begin{equation}\label{phieqbeta0bc}
-\at{\paren{a\nabla \phi+\tilde{\mb{b}}}\cdot \mb{n}}{\Gamma_{i,e}}
=\sum_{k=1}^N\alpha z_k(j_k+a_k).
\end{equation}
Suppose
\begin{equation}\label{currdownhill}
\PD{}{\jump{\phi}}\sum_{k=1}^Nz_kj_k>0.
\end{equation}
The above inequality states 
that current flowing out of the cell should increase if $\jump{\phi}$
increases, and is thus satisfied 
by biophysically reasonable expressions for $j_k$. 
This inequality is clearly satisfied if $j_k$ are of the form 
\eqref{HH} or \eqref{GHK} (see also \eqref{monotone}). 
Condition \eqref{currdownhill} is necessary for  
the boundary value problem \eqref{phieqbeta0} and \eqref{phieqbeta0bc}
to be uniquely solvable (up to an arbitrary constant), assuming 
$a_k$ is a given function of $\mb{x}$ (and $t$).

In connection to \eqref{phieqbeta0} and \eqref{phieqbeta0bc}, we 
perform the following calculation to illuminate the 
nature of system \eqref{beta0} as it relates to \eqref{dlessfull}.
Suppose $\phi$ satisfies \eqref{dlessfull}. 
By taking the time derivative of \eqref{dlesspoisson} with respect to $t$, 
we obtain:
\begin{equation}\label{phieq}
\begin{split}
\nabla \cdot 
0&=\paren{\beta^2\nabla\PD{\phi}{t}+a\nabla \phi+\widetilde{\mb{b}}},\\
a&=\sum_{k=1}^N z_k^2D_kc_k, \; \widetilde{\mb{b}}=
\sum_{k=1}^N \paren{-\text{Pe}z_kc_k\mb{u}+z_kD_k\nabla c_k},
\end{split}
\end{equation}
We used \eqref{dlessck} in deriving the above.
At the boundary, we may use \eqref{dlessbcpoisson} and \eqref{dlessckbc} 
to find that:
\begin{equation}\label{phieqbc}
\begin{split}
-&\at{\paren{\beta^2\nabla\PD{\phi}{t}+a\nabla \phi+\widetilde{\mb{b}}}\cdot \mb{n}}
{\Gamma_{i,e}}\\
=&\beta\theta\PD{}{t}\paren{C_m\jump{\phi}}
+\sum_{k=1}^N \paren{z_kc_k\PD{\mb{X}}{t}\cdot\mb{n}+\alpha z_k(j_k+a_k)}.
\end{split}
\end{equation}
If we formally let $\beta\to 0$ in \eqref{phieq} and \eqref{phieqbc}, 
we obtain \eqref{phieqbeta0} and \eqref{phieqbeta0bc} respectively.
For the above limit to be justified, we must require that $\PD{\phi}{t}$
and $\PD{[\phi]}{t}$ remain order $1$ with respect to $\beta$ as 
$\beta\to 0$. It is thus only when the evolution of $\phi$ and $[\phi]$
is sufficiently slow that we can reliably use system \eqref{beta0} as 
as an approximation to \eqref{dlessfull}. 

We see from 
\eqref{phieq} and \eqref{phieqbc} that there are two other time 
scales in the system besides the diffusive time scale $T_D$. 
The first is the Debye time scale, $\beta^2T_D$. This is the 
relaxation time scale of deviations from electroneutrality. 
This Debye time scale is too small to be of physiological interest, 
and we may safely ignore the $\beta^2\PD{\phi}{t}$ terms 
except for the very short initial layer that may exist 
depending on initial conditions. The other time scale,
$\beta\theta T_D$, which we shall call the cable time scale, 
is the time scale over which the membrane 
potential $[\phi]$ can change. In excitable tissue, the
ionic currents $j_k$ can change on a time scale comparable 
to $\beta\theta T_D$. The interaction of rapid changes
in $j_k$ and the capacitive current $\beta\theta\PD{}{t}(C_m[\phi])$
term lead to cable effects, including the propagation of 
action potentials, which is essential in describing 
a wide range of electrophysiological behavior. 
Such phenomena are usually described by the cable model, 
in which the intracellular and extracellular media are treated 
as ohmic resistive media \cite{HH,KS}.

Setting $\beta\theta\PD{}{t}(C_m[\phi])$ term to $0$ could thus be 
problematic for certain applications. It is thus of interest 
to develop a model in which the term $\beta^2\PD{\phi}{t}$ is ignored 
but $\beta\theta\PD{}{t}(C_m[\phi])$ is not, while retaining 
electrodiffusive and osmotic effects contained in the full model. 
Such a model (without osmotic effects) is proposed in 
\cite{Sinica}
and was applied in \cite{mori_ephaptic_2008} to a problem in cardiology.
A key ingredient in the derivation of such a model is 
an analysis of a boundary layer that forms at the membrane interface $\Gamma$.
This boundary layer, in physical terms, correspond to 
charge accumulation on both sides of the membrane.
We refer the reader to \cite{morithesis,mori2009three} for this model and 
its relationship to conventional cable models.
 
An important feature of system \eqref{beta0} is that 
it satisfies the following energy equality.
\begin{proposition}
Suppose $c_k, \phi, \mb{u}$ and $p$ 
are smooth functions that satisy system \eqref{beta0}.
Then, the following equality holds:
\begin{equation}
\D{}{t}\paren{G_S+E_\text{elas}}=-I_p-J_p-J_a.\label{ENFE}
\end{equation}
In the above, $G_S, E_\text{elas}, I_p, J_p$ and $J_a$
are dimensionless versions of corresponding quantities 
in \eqref{FE} and \eqref{FEE}.
\end{proposition} 
\begin{proof}
The proof follows from a calculation similar to the proof of Theorem
\ref{mainc}.
\end{proof}
Thus, system \eqref{beta0} may be seen as the system associated with 
the energy principle \eqref{ENFE} in which the electrostatic energy 
in \eqref{FEE} is discarded.

\section{Numerical Simulation of Animal Cell Volume Control}\label{animal}

In this section, we take the problem of 
cell volume control to illustrate some aspects of the model we introduced 
above. Cells contain a large number of organic molecules that do
not leak out through membrane. This results in excess 
intracellular osmotic pressure, 
which may cause the cell to burst. Cells have developed 
countermeasures to prevent this from happening.

We shall use the 
electroneutral system \eqref{beta0} to study cell volume 
control. 
We continue to work with the dimensionless equations.
To simplify matters, we suppose that 
the cell membrane $\Gamma$ and the outer boundary 
$\Gamma_{\text{out}}=\partial \Omega$ are concentric spheres
for all time 
and that the velocity field $\mb{u}$ only has a radial 
component. Assuming the boundary condition  
$\mb{u}=\mb{0}$ on $\Gamma_\text{out}$ 
we immediately see that $\mb{u}=\mb{0}$ throughout $\Omega_i\cup \Omega_e$.
We can thus drop equation \eqref{beta0stokes} and set $\mb{u}=\mb{0}$
wherever $\mb{u}$ appears in system \eqref{beta0}.  
Assuming further that $c_k$ and $\phi$ are functions only 
of the (dimensionless) radial coordinate $r$, we have:
\begin{subequations}\label{radial}
\begin{align}
\PD{c_k}{t}&
=-\frac{1}{r^2}\PD{}{r}\paren{r^2f_k}, \; 
f_k=-\paren{D_k\PD{c_k}{r}+z_kc_k\PD{\phi}{r}},\label{radialck}\\
\sum_{k=1}^{N}z_kc_k&=0,
\end{align}
for $0<r<R$ and $R<r<R_\text{out}$ where $R(t)$ is the radius
of the membrane sphere $\Gamma$ and $R_\text{out}=\text{const}.$
is a the radius of the outer boundary sphere $\Gamma_\text{out}$.
The boundary conditions are:
\begin{align}
f_k=&\begin{cases}
0 &\text{ at } r=0,\\
c_k\PD{R}{t}+\alpha(j_k+a_k) &\text{ at } r=R\pm,\label{radialckbc}
\end{cases}\\
-\frac{1}{\text{Pe}}\PD{R}{t}&=j_w, \; \jump{p}=F_\text{elas}
 \text{ at } r=R,\label{radialflow}
\end{align} 
\end{subequations}
where $r=R\pm$ denote limiting values as $r$ approaches $R$ from 
above or below. Boundary conditions at $R=R_\text{out}$, will 
be specified later. The elastic force $F_\text{elas}$ can now 
be viewed as a scalar quantity since the force is only in the 
radial direction.

We now develop a numerical algorithm to simulate 
system \eqref{radial} and apply this 
to animal cell volume control as a demonstrative example.

We first
discuss the numerical algorithm used to simulate system \eqref{radial}.
Consider \eqref{radial} in the region $a<r<b$. 
First, suppose $b<R$ or $a>R$. Then, we have:
\begin{equation}
\D{}{t}\int_a^b r^2c_kdr= a^2f_k(a)-b^2f_k(b).\label{fkab}
\end{equation}
If we let $b=R(t)$ in the above, we must account for 
the fact that $R(t)$ is changing in time. Using 
\eqref{radialck} and \eqref{radialckbc}, we have:
\begin{equation}\label{fkmem}
\D{}{t}\int_a^{R(t)}r^2c_kdr=a^2f_k(a)-R^2(t)\alpha(j_k+a_k).
\end{equation}
A similar expression is true when $a=R(t)$.
The above conservation relations will be the basis for 
our discretization. 

Let $\Delta t$ be the time step, and let $R^n$ be the position 
of the membrane at $t=n\Delta t$.
We divide $0<r<R^n$ and $R^n<r<R_\text{out}$ into $N_v$ equal segments.
Let
\begin{equation}
r_l^n=
\begin{cases}
\frac{kR^n}{N_v}, &\text{ if } 0\leq l\leq N_v\\
R^n+\frac{R_\text{out}-R^n}{N_v}, &\text{ if } N_v+1\leq l\leq 2N_v.
\end{cases}
\end{equation}
The $l$-th segment is given by $r_{l-1}<r<r_l$.
Of the $2N_v$ segments, segments $1\leq l\leq N_v$ are 
in the interior of the cell, whereas the the rest are in the 
exterior of the cell. 
In each segment, we have the concentrations 
$c_{k,l}^n$
and the electrostatic potential $\phi_l^n$.

Suppose we are to advance from time $(n-1)\Delta t$ to $n\Delta t$.
We use a splitting scheme.
Each time step is divided into two substeps. 
In the first substep, we advance membrane position:
\begin{equation}
R^n=R^{n-1}-\text{Pe}j_{w}^{n-1}\Delta t.
\end{equation}
In evaluating $j_w$, we need the osmotic pressure as 
well as the elastic force $F_\text{elas}$, both of which 
are evaluated using at time $(n-1)\Delta t$. For concentrations 
of ions at the intracellular and extracellular sides of the 
membrane, we use $c_{k,N_v}^{n-1}$ and $c_{k,N_v+1}^{n-1}$ respectively.

In the second substep, we update the concentrations and compute the 
electrostatic potential. We use one step of a backward Euler discretization.
We first describe our discretization for the intracellular region. 
Define:
\begin{equation}\label{ckinr}
c_{k,i}^n(r)=
\begin{cases}
c_{k,l}^n &\text{ if } r_{l-1}^n\leq r<r_l^n, \\
0 &\text{ if } r\geq r_{N_v}^n=R^n.
\end{cases}
\end{equation}
Suppose first that $R^n\leq R^{n-1}$. 
For $1\leq l\leq N_v-1$, we discretize \eqref{fkab} 
to obtain an equation for $c_{k,l}^n$:
\begin{equation}
\begin{split}
\frac{4\pi}{3}((r_l^n)^3-(r_{l-1}^n)^3)c_{k,l}^n&=
\int_{r_{l-1}^n}^{r_l^n} 4\pi r^2c_{k,i}^{n-1}(r)dr\\
&+4\pi \paren{(r_{l-1}^n)^2f_{k,l-1}^n-(r_l^n)^2f_{k,l}^n}\Delta t.\label{ckln}
\end{split}
\end{equation}
where $f_{k,l}^n$ is set to $0$ for $l=0$ and 
\begin{equation}
f_{k,l}^n=
-D_k\paren{\frac{c_{k,l}^n-c_{k,l-1}^n}{\Delta x_i}
+\frac{c_{k,l}^n+c_{k,l}^n}{2}\frac{\phi_{k,l}^n-\phi_{k,l-1}^n}{\Delta x_i}},
\text{ for } 1\leq l\leq N_v-1,
\end{equation}
where $\Delta x_i=R^n/N_v$. 
Note that the integral in \eqref{ckln} can be evaluated exactly 
given expression \eqref{ckinr}.
As for segment $l=N_v$, we view the endpoint $r_{N_v}^n=R^n$
as having evolved from $R^{n-1}$, and thus discretize \eqref{fkmem}.
We have:
\begin{equation}\label{ckmem}
\begin{split}
\frac{4\pi}{3}((r_{N_v-1}^n)^3-(R^n)^3)c_{k,N_v}^n&=
\int_{r_{N_v-1}^n}^{R^{n-1}} 4\pi r^2c_{k,i}^{n-1}(r)dr\\
&+4\pi \paren{(r_{N_v-1}^n)^2f_{k,N_v-1}^n-(R^n)^2\alpha(a_k^n+j_k^n)}\Delta t.
\end{split}
\end{equation}
The important point here is that the upper end point of the above 
integral is $R^{n-1}$ and not $R^n$.
The total membrane fluxes $(R^n)^2a_k^n$ and $(R^n)^2j_k^n$ 
are evaluated at time $n\Delta t$, 
and are thus functions of $c_{k,N_v}^n, c_{k,N_v+1}^n$ and 
$[\phi]^n=\phi_{N_v}^n-\phi_{N_v+1}^{n}$. 

If $R^n>R^{n-1}$, the discretized equations are  
the same as \eqref{ckln} and \eqref{ckmem} 
except that in \eqref{ckmem} the upper 
endpoint of the integral in is $R^n$ instead of $R^{n-1}$.
The fact that the endpoint of the integral is time-dependent in 
\eqref{fkmem} is taken into account by the $0$ extension of 
$c_{k,i}^{n-1}(r)$ when $r\geq R^{n-1}$ (see \eqref{ckinr}).  

The final equation we impose is that electroneutrality be satisfied 
in each segment:
\begin{equation}\label{discEN}
\sum_{k=1}^N z_kc_{k,l}^n=0 \text{ for all } l.
\end{equation}

For the extracellular segments $N_v+1\leq l\leq 2N_v$, we essentially use 
the same discretization as in the intracellular segments. 
The only difference is in treating boundary conditions at the 
$l=2N_v$ segment. We impose either no-flux or Dirichlet boundary 
conditions. For no-flux boundary conditions, we simply let 
$f_{k,N_v}^n=0$ in \eqref{ckln} for $l=2N_v$. Suppose the 
Dirichlet boundary conditions are given by:
\begin{equation}
c_k(R_\text{out},t)=c_{k,e}.
\end{equation}
In this case, we set:
\begin{equation}\label{discbc}
c_{k,e}^n=c_{k,e}, \; c_{k,e}^n
=\frac{3}{2}c_{k,2N_v}^n-\frac{1}{2}c_{k,2N_v-1}^n.
\end{equation}
For either boundary condition, the electrostatic potential 
is determined only up to an additive constant, and 
we thus set $\phi_e^n=3\phi_{2N_v}^n/2-\phi_{2N_v-1}^n/2=0$.

For the second substep, we thus have equations \eqref{ckln},
\eqref{ckmem} and \eqref{discEN} with suitable boundary conditions 
at $r^n_{2N_v}=R_\text{out}$, which we must solve for $c_{k,l}^n$
and $\phi_l^n$. This system of nonlinear algebraic equations
is solved using a Newton iteration where the Jacobian matrix 
is computed analytically. In all simulations reported here, 
we obtained convergence to within a relative tolerance of $10^{-12}$
within less than $4$ iterations. In particular, the electroneutrality 
condition at each time step was satisfied at each point to within 
$6\times 10^{-14}$mmol/$\ell$ for all simulation results shown below. 

Note that the discretization is conservative. For example, we have:
\begin{equation}
\sum_{l=1}^{2N_v} \frac{4\pi}{3}((r_l^n)^3-(r_{l-1}^n)^3)c_{k,l}^n=\text{const}
\end{equation}
so long as we impose the no-flux boundary condition at $r=R^{\text{out}}$.
We have checked this property numerically, we achieve conservation 
of ions to $14$ to $15$ digits. 
This property is very important in studying long time behavior. 

We would also like to comment on our use of the 
backward Euler scheme and the Newton iteration in the second 
substep of each time step. Rather than use a backward Euler step, 
we may split the second substep further 
into two substeps. In the first substep, one compute 
the updates of $\phi$ given values of $c_k$ at time $(n-1)\Delta t$
and in the second substep, we update $c_k$ using the 
updated $\phi$. 
A variant of this scheme is to use the above as one step 
of a fixed point iteration to solve the backward Euler problem.
An advantage of these schemes is that  
the associated matrix problem is much simpler and smaller 
than that of a full Newton iteration we use in this paper. 
This was indeed the 
first algorithm we used in our attempt to simulate the system. 
This algorithm, however, turned out to have serious stability 
and convergence issues and led to large pile-up of charges close to the 
membrane. This difficulty was clearly caused by the moving membrane.
Indeed, a similar algorithm was successfully used 
in \cite{CAMCoS} to simulate a similar 
but higher dimensional system, 
in which the membrane was stationary. We also found that if 
$\Delta t$ or the membrane velocity is very small, the 
fixed-point algorithm does produce computational results 
in agreement with those obtained using a Newton iteration.  
We do point out that even the backward Euler, Newton scheme, 
that we use here was not unconditionally stable, though the 
time step restriction was never serious. 
A more stable algorithm may be possible by 
developing a scheme in which the membrane position and 
concentrations (and electrostatic potential) are computed 
simultaneously. 

We now describe the model example we simulate.
The cell membrane of animal cells is not mechanically strong enough to resist 
osmotic pressure due to the presence of organic solutes 
in the cell. Cell volume control is achieved by actively maintaining 
a concentration gradient of ions across the cell membrane. 
Many modeling studies have been performed to study cell volume 
control in animal cells. To the best of our knowledge, all 
such studies use ODE systems in which the cellular and extracellular 
concentrations are assumed to have no spatial variation
\cite{KS,hoppensteadt2002modeling,tosteson1960regulation,tosteson1964regulation,jakobsson1980interactions}.
The novelty here is that we  
use the PDE system \eqref{radial}, a field theory,
to study cell volume control.
 
We consider 
a generic spherical animal cell whose sodium and potassium 
concentration differences across the membrane is maintained 
by the presence of the Na-K ATPase. 
Henceforth, we shall use variables with their original 
dimensions, since we will be dealing with a concrete 
biophysical setup.
We consider four 
species of ion, Na$^+$, K$^+$, Cl$^-$ and the organic 
anions, which we index as $k=1,\cdots, 4$ in this order. 
The diffusion coefficients of the four species is 
given in the Table \ref{ionparams}. We make the simplification that  
the organic anions are a homogeneous species with a single diffusion 
coefficient. The diffusion coefficient for 
the organic anion is somewhat arbitrary, one order of 
magnitude smaller than the small inorganic ions. 

We take the initial radius of the spherical cell to be 
$R_0$. We let the outer edge of the simulation domain 
$R_\text{out}=2R_0$.  
We assume that the membrane does not generate any mechanical 
force, so that $F_\text{elas}=0$. 
Passive water flow across the membrane 
is proportional to the water chemical potential. 
Given that $F_\text{elas}=0$, water 
flow across the cell membrane is driven 
osmotic pressure difference across the membrane:
\begin{equation}
j_w=\zeta N_Ak_BT\sum_{k=1}^4 \jump{c_k}
\end{equation}
where $N_A$ is the Avogadro constant (so that $N_Ak_B$ is the ideal 
gas constant), $T$ is the absolute temperature, 
$c_k$ is measured in mmol$/\ell$ 
and $\zeta$ is measured in velocity per pressure. 

For the passive membrane flux $j_k$, we take expression 
\eqref{GHK}:
\begin{equation}
j_k=\frac{R_0^2}{R(t)^2}P_kz_k\phi'\paren{
\frac{\at{c_k}{R-}\exp(z_k\phi')-\at{c_k}{R+}}
{\exp(z_k\phi')-1}}, \; \phi'=\frac{q\jump{\phi}}{k_BT}.
\end{equation}
where the subscript $R-$ and $R+$ denote evaluation at the 
inner and outer faces of the membrane. 
This choice is standard for cell volume studies 
\cite{strieter_volume-activated_1990,lew_behaviour_1979}.
The number $P_k$ is measured in cm/sec and is the permeability 
of a unit area of membrane for ionic species $k$ when the radius 
of the cell is $R_0$. Assuming that this permeability is 
determined by the presence of ionic channels and 
that the number of ionic channels remains constant, $j_k$
must be made inversely proportional to the membrane area.  
For sodium, potassium and chloride, $P_k$ is positive 
but we set the permeability for organic solutes to $0$.

We follow \cite{lew_behaviour_1979} to use the following expression for 
the Na-K ATPase flux:
\begin{equation}
a_1=A_p\paren{\frac{\at{c_1}{R-}}{\at{c_1}{R-}+K_{Na}}}^3
\paren{\frac{\at{c_2}{R+}}{\at{c_2}{R+}+K_K}}^2, \; a_2=-\frac{2}{3}a_1.
\end{equation}
Recall here that $a_1, c_1$ are the active Na$^+$
flux and concentration respectively and $a_2, c_2$ are the 
active K$^+$ flux and concentration respectively. 
The exponents of $3, 2$ and 
the factor of $-2/3$ reflects the $3:2$ stoichiometry of the 
NaK ATPase in pumping Na$^+$ out and K$^+$ into the cell. 
The constants $K_{Na}$ and $K_k$ are given in Table \ref{miscparams}.

All constants and initial conditions are given in Tables \ref{miscparams}
and \ref{ionparams}.
Initial concentrations are assumed spatially uniform. 
The constants that are not listed in the tables are computed 
so that the initial state is a stationary state under no-flux 
boundary conditions at $R=R_\text{out}$. 
This is similar to what is done in \cite{lew_behaviour_1979}.
This procedure determines the initial intracellular 
Cl$^-$ concentration, Na-K ATPase maximal pump rate $A_p$, 
K$^+$ permeability $p_2$, initial intracellular 
organic solute concentration,
and the organic solute valence $z_4$. We point out 
that $\jump{\phi}^\text{init}$, the initial value of the 
membrane potential is only needed to compute the 
initial conditions. Once all the concentrations are 
known, the concentrations serve as the initial conditions 
and there is no need to know $[\phi]$ at the initial 
time to evolve the system forward. 

\begin{table}
\begin{center}
\begin{tabular}{|c|c||c|c|}
\hline
$T$ (K)& $273.15+37$ & $K_k$ (mmol/$\ell$)& $0.75$ \cite{lew_behaviour_1979}\\
$\zeta$ (cm/s/mPa) & $5.2507\times 10^{-13}$ \cite{strieter_volume-activated_1990}
&$K_{Na}$ (mmol/$\ell)$ & $3.5$ \cite{lew_behaviour_1979}\\
$R_0$ (mm)& $0.5$ & $A_p$ (cm/s) & -\\
$R_\text{out}$ (mm)& $1$ & $\jump{\phi}^\text{init}$ (mV)& $-70$ \\
\hline
\end{tabular}
\end{center}
\caption{Constants used in the numerical simulation. 
$\jump{\phi}^\text{init}$ is the initial membrane voltage. 
Symbols labeled 
with '-' are determined so that the initial 
condition is a stationary state (see main text). 
The ion related 
constants are listed in Table \ref{ionparams}.}
\label{miscparams}
\end{table}

\begin{table}
\begin{center}
\begin{tabular}{|c||c|c|c|c|c|}
\hline
      & $z_k$ & $D_k$ (cm$^2$/s) & $P_k$ (cm/s)
& $c_{k,\text{int}}^\text{init}$ & $c_{k,\text{ext}}^\text{init}$\\
\hline
Na$^+$&$+1$&$1.33\times 10^{-5}$\cite{koch_biophysics_1999}&$1.0\times 10^{-7}$\cite{hernandez_modeling_1998}&$10$&$145$\\
K$^+$ &$+1$&$1.96\times 10^{-5}$\cite{koch_biophysics_1999}&  -                  &$140$&$5$\\
Cl$^-$&$-1$&$2.03\times 10^{-5}$\cite{koch_biophysics_1999}&$1.0\times 10^{-7}$\cite{hernandez_modeling_1998}&-&$150$ \\
O.A.  &-& $1.0\times 10^{-6}$&  0                  &-&$0$\\
\hline
\end{tabular}
\end{center}
\caption{Parameters related to ionic concentrations. O.A. stands 
for organic anions.  
The initial intracellular and extracellular concentrations are 
given by $c_{k,\text{int}}^\text{init}$ and $c_{k,\text{ext}}^\text{init}$
respectively (listed here in mmol/$\ell$). Symbols labeled 
with '-' are determined so that the initial 
condition is a stationary state (see main text). Other parameters 
are listed in Table \ref{miscparams}.}
\label{ionparams}
\end{table}

In the simulations to follow, we took $N_v=100$ and the time step 
$\Delta t=500$ms.

We perform the following numerical experiments. 
Starting with the initial conditions specified above 
with no-flux boundary conditions, we set the following 
Dirichlet boundary conditions for $t\geq 10$s:
\begin{equation}
c_{1,e}=100, \; c_{2,e}=50, \; c_{3,e}=150,\label{hKcb}
\end{equation}
where the units are in mmol/$\ell$. The boundary concentrations 
are thus isotonic with the initial concentrations, but 
the extracellular K$^+$ concentration is now increased 
$10$-fold. Such a stimulus should lead to immediate depolarization 
together with expansion of the cell. 

The computational results are given in Figure \ref{highK}.
What is interesting here is that there is a transient 
drop in the cell radius, followed by an expected gradual 
increase. This transient drop is due to the following. 
After a sudden change in the boundary condition,  
Na$^+$ ions diffuse out whereas K$^+$ ions should diffuse 
in from $R=R_\text{out}$. Since K$^+$ diffuses faster 
than Na$^+$, there is a transient increase in total 
ionic concentration near the membrane, leading to 
excess osmotic pressure immediately outside the cell
compared to the inside. This gives rise to a transient 
drop in the cell radius. However, as the ionic concentration 
becomes spatially uniform within the extracellular and intracellular 
domains, the cell starts to expand. 

\begin{figure}
\begin{centering}
\includegraphics[width=\textwidth]{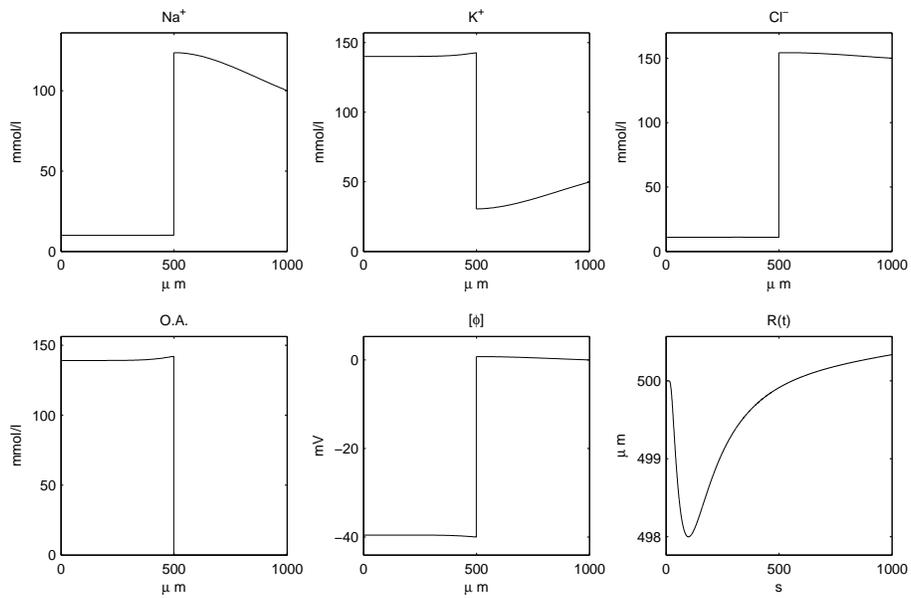}
\end{centering}
\caption{Computational results under a high K$^+$ stimulus
(see \eqref{hKcb}). 
The first five figures are the snapshots of the ionic concentrations 
and the electrostatic potential at $t=50$s. The horizontal axis 
represents the radius $r$. The last figure plots the cell radius $R(t)$
as a function of time.}
\label{highK}
\end{figure}

The next computational results describe a hypotonic 
shock. We set the boundary conditions to the following 
for $t\geq 10$s:
\begin{equation}
c_{1,e}=100, \; c_{2,e}=5, \; c_{3,e}=105, \label{hypcb}
\end{equation}
where the concentrations are in mmol/$\ell$. 
A snapshot of the computational results are given in Figure \ref{hypotonic}.
The cell expands due to the hypotonic shock but tends to 
a new stationary state with time. 

\begin{figure}
\begin{centering}
\includegraphics[width=\textwidth]{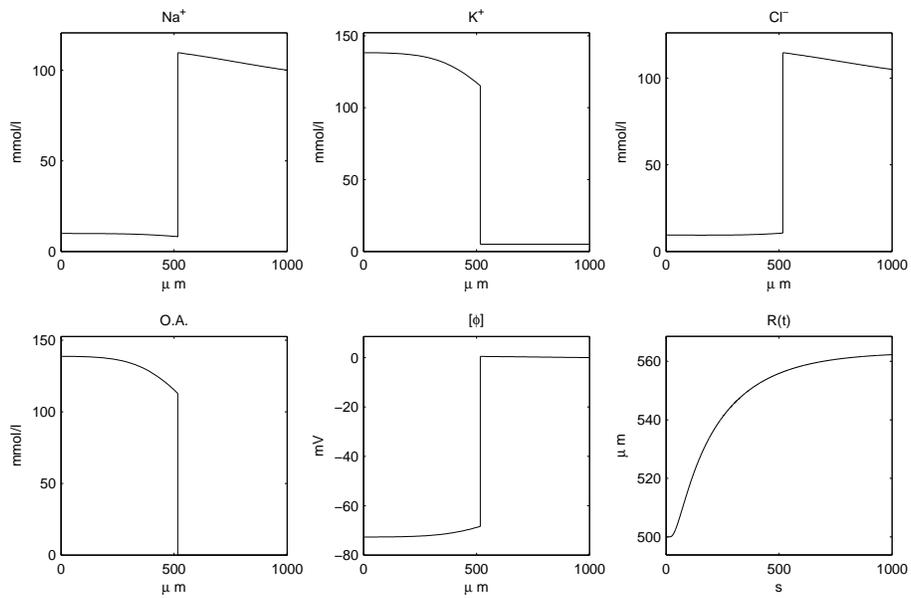}
\end{centering}
\caption{Computational results under hypotonic stimulus (see \eqref{hypcb}). 
The first five figures are the snapshots of the ionic concentrations 
and the electrostatic potential at $t=100$s. The horizontal axis 
represents the radius $r$. The last figure plots the cell radius $R(t)$
as a function of time.}
\label{hypotonic}
\end{figure}

\section{Conclusion}\label{conclusion}

We introduced a PDE system of electrodiffusion and 
osmotic water flow in the presence of deformable capacitance-carrying 
membranes. The salient feature of the model is that it satisfies 
an energy equality, and thus possesses a natural thermodynamic 
structure. We discussed simplifications of the model and 
applied the electroneutral limit to the problem of cell volume 
control.

In the proof of Theorem \ref{mainc}, we showed 
that the van t'Hoff expression 
for osmotic pressure arises naturally, simply through an integration by 
parts argument. This observation seems to be new. It is 
interesting that, in expression \eqref{FE}, the mechanical 
pressure $p$ and osmotic pressure $\pi_w$ only appear 
in the combination $\psi_w=p+\pi_w$. 
This is consistent with experimental results indicating that
the effect of osmotic pressure on 
transmembrane water flow is indistinguishable from 
that of mechanical pressure \cite{finkelstein1987water}.

The models introduced here are {\em sharp interface} models 
in the sense that the membrane is treated as a surface 
without thickness and the physical quantities of interest 
are allowed to have discontinuities across $\Gamma$. 
This is in contrast to {\em diffuse interface} models in 
which the membrane has some small but finite thickness 
and the physical quantities transition rapidly but smoothly 
across the interface.
It should be possible to obtain at least parts of the model by 
taking the thin interface limit of an appropriate 
diffuse interface (or finite thickness) model. 
This may lead to a simpler verification of the 
energy identity of Theorem \ref{main}. Establishing 
such a connection may also help in understanding the physical nature 
of the capacitive force \eqref{FcapCmQ}. The calculations 
in Appendix \ref{capforce} may be seen as an initial 
step in establishing this relationship.

Given the natural thermodynamic structure of the problem, 
it is almost certainly the case that our model has a {\em variational}
structure. A variational principle for dissipative systems 
dates back to \cite{onsager1931reciprocal}. This procedure 
has been used successfully in deriving dynamic equations for 
soft matter systems \cite{doi1988theory,doi2009gel}.
In \cite{eisenberg2010energy,hyon2010energetic}, a 
model for non-ideal electrolyte solutions is derived
by combining, in the spirit of Rayleigh (see \cite{goldstein1980classical}), 
the principle of least action with the above variational 
principle for dissipative system.

The energy identities introduced here provide a natural 
apriori estimate for our PDE system. For related systems, 
the corresponding 
free energy identity has been used successfully to prove 
stability of steady states \cite{biler2000long,ryham2009existence}. 
It would be interesting to 
see if similar analytical results can be obtained for the 
system proposed here. 

We hope that our model has wide-ranging applications in 
cellular physiology. 
In principle, our model is applicable to 
most problems of classical physiology 
\cite{davson1970textbook,boron2008medical,pappenheimer1987silver}.
As we saw in Section \ref{simple}, our 
model admits simplifications when certain dimensionless 
parameters are small. In the short time scale when water movement 
is not significant, the system is reduced to the
Poisson-Nernst-Planck model with interface boundary conditions.
This and related models have been successfully applied in \cite{
leonetti_biomembrane_1998,leonetti_pattern_2004,mori_ephaptic_2008}. 
If the physiological processes of interest are slow and happen 
over a long time scale, the electroneutral limit may be taken. 
This was applied to the problem of cell volume control in 
Section \ref{animal} of this paper. 

Any serious application of our model will require the development 
of an efficient numerical algorithm. 
The electrodiffusive part of the problem with stationary membranes 
(without fluid flow)
has been treated successfully in \cite{CAMCoS, brera2010conservative}
in a two-dimensional setting. 
In the model presented here, the membrane interface is dynamic.
We must therefore solve an electrodiffusive problem in a domain 
with a moving interface across which physical quantities 
experience discontinuities. 
We have presented successful computations in one-dimension for 
the electroneutral limit in Section \ref{animal}, 
but simulations are bound to be more challenging in higher dimension.
If a regular mesh is to be used, 
immersed boundary or immersed interface schemes could be a major 
component of the algorithm 
\cite{layton2006modeling,ibmethod,li_immersed_2006}.

Many physiological phenomena in which both electrodiffusion and osmosis 
play an important role take place over spatial scales of whole tissues
or organs rather than the cellular spatial scale we focused on in this 
paper. Such systems include ocular fluid circulation, electrolyte 
regulation in the kidney or brain ionic homeostasis. 
For such systems, it is important to develop an appropriate 
homogenized model. In the context of cable models, this is known 
as the bidomain model, and has found great utility in many 
contexts, especially in cardiac electrophysiology 
\cite{KS,eisenberg_three-dimensional_1970,eisenberg1979electrical,
mathias1979electrical,neu1993hst}. We shall report on a such 
a multidomain model in a future publication.

\appendix
\section{Appendix}

\subsection{Physical Interpretation of the Capacitive Force}\label{capforce}

Let the membrane be made of an incompressible material. 
In this case, we argued that $C_m$ should satisfy \eqref{Cmscaling}.
We shall show that the incompressibility of 
the material implies $\tau_\text{cap}=-C_m[\phi]^2$, in 
agreement with \eqref{FcapCmQ}. To this end, we take
a membrane of finite thickness $d$ and consider the limit as 
the thickness tends to $0$. Let 
$\Gamma$ be the midplane of this membrane of finite thickness. 
The membrane thus coincides with $\Gamma$ as $d\to 0$.

Take a point $\mb{x}\in\Gamma$ and let $\mb{n}$ be the unit normal 
and $d$ the thickness of the membrane at $\mb{x}$. 
The stress inside the membrane is given by:
\begin{equation}\label{sigmamem}
\Sigma^\text{mem}=\Sigma_e^\text{mem}
-p^\text{mem}I
\end{equation}
where $\Sigma_e^\text{mem}$ is the Maxwell stress,
$\Sigma_\text{elas}^\text{mem}$ the elastic stress $p^\text{mem} I$
is the isotropic pressure term that enforces incompressibility
of the material. We have made the simplification that the material 
can only generate isotropic stresses.  

Let us now consider the limiting behavior of 
this stress when $d$ is very small.
To leading order in $d$, 
the Maxwell stress tensor inside the membrane is given by:
\begin{equation}\label{sigmaemem}
\Sigma_e^{\text{mem}}=
\epsilon_m\frac{[\phi]^2}{d^2}\paren{\mb{n}\otimes\mb{n}-\frac{1}{2}I}
=C_m\frac{[\phi]^2}{d}\paren{\mb{n}\otimes\mb{n}-\frac{1}{2}I}
\end{equation}
where $\epsilon_m$ is the dielectric constant of the membrane.
We assumed that the electric field inside the membrane 
is given by $[\phi]\mb{n}/d$, given that the membrane is very thin.
We used $\epsilon_m/d=C_m$ in the second equality.

Now, let us consider stress balance at $\mb{x}+\frac{d}{2}\mb{n}$, 
the point where the membrane touches $\Omega_e$. Here, we have 
the following stress balance condition:
\begin{equation}
\Sigma^\text{mem}\mb{n}=\Sigma^{\Omega_e}\mb{n}
\end{equation}
where $\Sigma^{\Omega_e}$ is the stress in $\Omega_e$. 
Using \eqref{sigmamem} and \eqref{sigmaemem}
$\Sigma^\text{mem}\mb{n}$, to leading order in $d$, can be written as:
\begin{equation}\label{sigmamemexp}
\Sigma^\text{mem}\mb{n}=
\paren{C_m\frac{[\phi]^2}{2d}-p^\text{mem}}\mb{n}.
\end{equation}
As $d\to 0$, 
$\Sigma^{\Omega_e}\mb{n}$ must remain finite if there is a 
finite distinguished limit. Therefore, $\Sigma^\text{mem}\mb{n}$
must remain order $1$ with respect to $d$. 
In \eqref{sigmamemexp}, the term $C_m\frac{[\phi]^2}{2d}$
grows like $1/d$ as $d\to 0$. The elastic stress stays order $1$
with respect to $d$. Therefore, $p^\text{mem}\mb{n}$, must satisfy:
\begin{equation}
p^\text{mem}=C_m\frac{[\phi]^2}{2d}
\end{equation} 
to leading order in $d$.

Take any unit vector $\mb{t}$ tangent to $\Gamma$ at $\mb{x}$. We have:
\begin{equation}
\Sigma^{\text{mem}}\mb{t}=-\frac{1}{d}C_m[\phi]^2\mb{t},
\end{equation}
to leading order in $d$.
Multiplying the above by the thickness $d$ of the membrane, and 
taking the limit as $d\to 0$, we conclude
that a ``surface tension'' of magnitude 
$-C_m[\phi]^2$ is generated at the membrane. 

The above derivation suggests the following physical interpretation
of expression \eqref{FcapCmQ}. The term $\frac{1}{2}C_m[\phi]^2$ 
comes directly from the Maxwell stress. More simply put, 
this tension comes from large coulomb forces squeezing the thin membrane. 
This force must be counter balanced to maintain the mechanical 
integrity of the membrane, which is given by an isotropic 
pressure in the case of incompressible materials. 
This contributes the term 
$\frac{1}{2}Q\PD{C_m}{Q}[\phi]^2$ to the capacitive force.

\subsection{Relation to an Immersed Boundary Model of Osmosis}\label{leeatz}

In this appendix, we explore a model of osmosis proposed in 
\cite{lee2008immersed,atzberger2009microfluidic}
and compare this with the model presented in this paper.

Consider the system of equations introduced in Section \ref{diffosm}.
We shall take $\omega=\omega_0$ given in \eqref{ent}.
Let there be only one solute species and suppose that it is 
impermeable to the membrane. Let $c$ be the concentration of this 
solute. We have:
\begin{subequations}\label{appsys}
\begin{align}\label{appeq}
\PD{c}{t}+\nabla \cdot(\mb{u}c)&=\nabla \cdot \paren{D\paren{\nabla c}},\\
\nu \Delta \mb{u}-\nabla p&=0,
\end{align}
and at the boundary, we have:
\begin{equation}\label{appbc}
\jump{\Sigma(\mb{u},p)}=\mb{F}_\text{elas}, \; 
\mb{u}c-D\PD{c}{\mb{n}}=c\PD{\mb{X}}{t}\cdot \mb{n}
\end{equation}
\end{subequations}
and \eqref{cont}. We impose no-flux and no-flow boundary conditions 
on $\partial \Omega$. 
We know from Theorem \ref{mainc} that the solutions to 
\eqref{appsys} satisfy:
\begin{equation}
\begin{split}\label{appen}
&\D{}{t}\int_\Omega \paren{k_BT c\ln c+E_\text{elas}}d\mb{x}\\
=&-\int_\Omega Dc\abs{\nabla \ln c}^2d\mb{x}
+\int_\Gamma \paren{\mb{F}_\text{elas}\cdot \mb{n}+k_BT[c]}j_wdm_\Gamma.
\end{split}
\end{equation}
Dissipation in the free energy will be guaranteed if, for example, 
$j_w$ is a linear function of 
$\mb{F}_\text{elas}\cdot \mb{n}+k_BT[c]$ of negative slope.

In \cite{lee2008immersed,atzberger2009microfluidic}, the authors propose the following immersed boundary model 
of osmosis. 
We argue why this may be considered a regularization 
of the above model.
The concentration $c$ satisfies:
\begin{subequations}\label{sys_pl}
\begin{equation}\label{addiff_pl}
\PD{c}{t}+\nabla \cdot(\mb{u}c)=
\nabla \cdot \paren{D\paren{\nabla c+\frac{c}{k_BT}\nabla\psi_\eta}}.
\end{equation}
Here, $\psi_\eta$ is a ``barrier'' potential localized near the membrane:
\begin{align}
\psi_\eta(\mb{x})&=\int_{\Gamma_\text{ref}} 
\varphi_\eta(\mb{x}-\mb{X}(\bm{\theta},t))dm_{\Gamma_\text{ref}}\\
\varphi_\eta(\mb{x})&=\eta^{-3}\varphi_1\paren{\frac{\mb{x}}{\eta}}, 
\end{align}
where $\varphi_1(\mb{x})$ is a positive  
function of compact support. 
For simplicity, we shall take $\varphi(\mb{x})$ to 
be radially symmetric and nonzero only if $\mb{x}\leq 1$. 
The parameter $\eta>0$ is the width of the barrier potential.  
As $\eta\to 0$, one would expect the 
barrier to be so high that the solute will be effectively blocked
from crossing the membrane. 

The fluid velocity field satisfies the 
Stokes equation in $\Omega\backslash \Gamma$ with an external 
force term that comes from the presence of solutes interacting 
with the barrier:
\begin{equation}\label{stokes_pl}
\nu \Delta \mb{u}-\nabla p= c\nabla \psi_\eta.
\end{equation}
At the boundary, we suppose that the velocity field is 
continuous and that the stress tensor $\Sigma_m(\mb{u},p)$
satisfies the following jump condition:
\begin{align}\label{stokes_plbc}
\jump{\Sigma_m(\mb{u},p)\mb{n}}&=\mb{F}_\text{elas}+\mb{F}_\text{sol},\\
\mb{F}_\text{sol}\abs{\PD{\mb{X}}{\bm{\theta}}}
&=\int_\Omega c\nabla \varphi_\eta(\mb{x}-\mb{X}(\bm{\theta},t))d\mb{x}.
\end{align}
\end{subequations}
The force $\mb{F}_\text{sol}$ should be seen as a reaction 
force to the $c\nabla \psi$ term in \eqref{stokes_pl}.
The membrane velocity $\PD{\mb{X}}{t}$ and $\mb{u}$ satisfy
relation \eqref{cont}. The following energy identity is satisfied by
this system:
\begin{equation}
\begin{split}\label{en_pl}
&\D{}{t}\int_\Omega \paren{k_BT c\ln c+E_\text{elas}+c\psi_\eta}d\mb{x}\\
=&-\int_\Omega Dc\abs{\nabla \paren{\ln c+\frac{\psi_\eta}{k_BT}}}^2d\mb{x}
+\int_\Gamma \paren{\mb{F}_\text{elas}+\mb{F}_\text{sol}}
\cdot \mb{n}j_wdm_\Gamma.
\end{split}
\end{equation}
We leave the verification of this identity to the reader.
The above energy will be monotonically decreasing if $j_w$ is a 
linear function of $\mb{F}_\text{elas}+\mb{F}_\text{sol}$ of negative slope.

Let us now consider the limit as $\eta\to 0$. 
First, consider the advection-diffusion equation satisfied by the solute.
The $\eta \to 0$ limit corresponds to the barrier potential 
becoming infinitely thin and high, and thus, this limit should lead to 
\eqref{appeq} in $\Omega\backslash \Gamma$ and 
the impermeable boundary condition \eqref{appbc} for the solute.
 
To examine the limiting behavior of the 
Stokes equation rewrite \eqref{stokes_pl} and \eqref{stokes_plbc}
in the following {\em immersed boundary} fashion:
\begin{equation}\label{stokes_plib}
\nu \Delta \mb{u}-\nabla p= c\nabla \psi_\eta
+\mb{f}_\text{sol}+\mb{f}_\text{elas}.
\end{equation}
Here $\mb{f}=\mb{f}_\text{sol}$ or $\mb{f}_\text{elas}$ are surface measures 
supported on $\Gamma$ whose action on a test function $\mb{v}$ is given by:
\begin{equation}
\dual{\mb{f}}{\mb{v}}=
\int_\Gamma \mb{F}(\mb{X})\cdot \mb{v}(\mb{X})dm_\Gamma
\end{equation}
where $\mb{F}=\mb{F}_\text{sol}$ or $\mb{F}_\text{elas}$. 
Equation \eqref{stokes_plib} is thus to be understood in the sense 
of distributions. The interface conditions \eqref{stokes_plbc}
are now expressed in terms of singular force fields. This 
reformulation of interface boundary conditions is a key 
ingredient in the immersed boundary method \cite{ibmethod}.

Given that $c\nabla \psi_\eta$ and $\mb{f}_\text{sol}$
are forces of action and reaction, as $\eta\to 0$, one expects 
$c\nabla \psi+\mb{f}_\text{sol}$ to go to $0$ in a distributional sense. 
The right hand side of \eqref{stokes_plib} will thus be reduced to 
$\mb{f}_\text{elas}$. Equation \eqref{stokes_plib} will then 
be equivalent imposing the boundary condition \eqref{appbc} 
at $\Gamma$ to the Stokes equation satisfied in $\Omega\backslash\Gamma$. 

In view of \eqref{appen}, \eqref{en_pl} and the discussion 
after these equations, we would like to 
show that $\mb{F}_\text{sol}$ approaches $[c]k_BT$, the van t'Hoff 
expression of osmotic pressure, as $\eta\to 0$. If this is the case, 
we can view system \eqref{sys_pl} as giving a mechanical interpretation
of osmotic pressure. In this view, osmotic pressure is generated 
by solute molecules hitting the impermeable membrane. We now sketch 
a boundary layer calculation that lends credence to this claim. 
Take a point on $\Gamma$ and let this point be the origin without 
loss of generality. Assume $\Gamma$ is locally flat. Let $\Gamma$ 
coincide with the $x-y$ plane and let $z$ be the direction 
normal to $\Gamma$ the positive $z$ axis to be pointing 
into $\Omega_e$. Assume furthermore, that 
$\abs{\PD{\mb{X}}{\bm{\theta}}}=\text{const}$ over this flat region.
Consider the solute flux for $z>0$:
\begin{equation}
J_z=\PD{c}{z}+\frac{c}{k_BT}\PD{\psi_\eta}{z}
\end{equation}
Introduce the stretched boundary layer coordinate $\eta Z=z$:
\begin{equation}
\eta J_z=\PD{c}{Z}+\frac{c}{k_BT}\PD{\psi_\eta}{Z}
\end{equation}
We conclude, to leading order, that
\begin{equation}
c(Z)=c(Z=1)\exp(-\psi_\eta(Z)/k_BT)+o(1).
\end{equation}
In original coordinates, we have, for $0<z<\eta$,
\begin{equation}
c(\mb{x}',z)=c(\mb{x}',\eta)\exp\paren{-\frac{\psi_\eta(z)}{k_BT}}+o(1)
\end{equation}
where $\mb{x}'=(x,y)$ and  
$o(1)$ denotes terms that tend to $0$ as $\eta \to 0$.
Likewise, for $-\eta<z<0$, we have:
\begin{equation}
c(\mb{x}',z)=c(\mb{x}',-\eta)\exp\paren{-\frac{\psi_\eta(z)}{k_BT}}+o(1).
\end{equation}
Let us now compute $\mb{F}_\text{sol}$ at $x=y=0$.
\begin{equation}
\begin{split}
\mb{F}_\text{sol}
\abs{\PD{\mb{X}}{\bm{\theta}}}&
=\int_{\abs{\mb{x}}\leq \eta}c\nabla \varphi_\eta(\mb{x})d\mb{x}\\
&=\int_{\abs{\mb{x}}\leq \eta,z\geq 0}c\nabla \varphi_\eta(\mb{x})d\mb{x}
+\int_{\abs{\mb{x}}\leq \eta,z\leq 0}c\nabla \varphi_\eta(\mb{x})d\mb{x}\\
&\equiv \paren{\mb{F}^+_\text{sol}+\mb{F}^-_\text{sol}}
\abs{\PD{\mb{X}}{\bm{\theta}}}.
\end{split}
\end{equation}
where we used the fact that $\varphi_\eta$ is supported in 
$\abs{\mb{x}}\leq \eta$. 
Consider the first term:
\begin{equation}
\begin{split}
\mb{F}^+_\text{sol}
\abs{\PD{\mb{X}}{\bm{\theta}}}
=&\int_{\abs{\mb{x}}\leq \eta,z\geq 0}c\nabla \varphi_\eta(\mb{x})d\mb{x}\\
=&
\int_{\abs{\mb{x}}\leq \eta,z\geq 0}\paren{c(\mb{x}',\eta)
\exp\paren{-\frac{\psi_\eta(z)}{k_BT}}\nabla \varphi_\eta(\mb{x})}d\mb{x}+o(1)\\
=&\int_{\abs{\mb{x}}\leq \eta,z\geq 0}\paren{c(\mb{x}',\eta)
\exp\paren{-\frac{\psi_\eta(z)}{k_BT}}
\PD{}{z}\varphi_\eta(\mb{x}',z)}d\mb{x'}dz\mb{e}_z
+o(1)
\end{split}
\end{equation}
where $\mb{e}_z$ is the unit coordinate vector in the $z$ direction.
Note that:
\begin{equation}
\psi_\eta=\abs{\PD{\mb{X}}{\bm{\theta}}}^{-1}
\int_{\abs{\mb{y}'}\leq \eta}\varphi_\eta(\mb{y}',z)d\mb{y}'
\end{equation}
near the origin. Here, we used the fact that the membrane is flat
and that $\abs{\PD{\mb{X}}{\theta}}$ is constant.
Using this, we have:
\begin{equation}
\begin{split}
\mb{F}^+_\text{sol}&=
\int_{z\geq 0}\paren{c(0,\eta)
\exp\paren{-\frac{\psi_\eta(z)}{k_BT}}\PD{}{z}\psi_\eta(z)}dz\mb{e}_z
+o(1)\\
&=\paren{-c(0,\eta)k_BT+c(0,0)k_BT\exp\paren{-\frac{\psi_\eta(0)}{k_BT}}}
\mb{e}_z+o(1).
\end{split}
\end{equation}
Likewise, we have:
\begin{equation}
\mb{F}^-_\text{sol}=
\paren{c(0,-\eta)k_BT-c(0,0)k_BT\exp\paren{-\frac{\psi_\eta(0)}{k_BT}}}
\mb{e}_z+o(1).
\end{equation}
Thus, as $\eta \to 0$, we have:
\begin{equation}
\mb{F}_\text{sol}=(c(0,0-)-c(0,0+))k_BT\mb{e}_z
\end{equation}
as desired.

We emphasize that the above calculation is purely symbolic.
We are not attempting to prove that the solutions
to \eqref{sys_pl} approach the solutions to \eqref{appsys}
in the $\eta\to 0$ limit.

\vspace{5mm}

\section*{Acknowledgments} 
The authors gratefully acknowledge support from the following sources:
National Science Foundation (NSF) grant DMS-0914963,
the Alfred P. Sloan Foundation and the McKnight Foundation (to YM),
NSF grant DMS-0707594 (to CL), and National Institutes of Health 
grant GM076013 (to RSE). The authors are grateful to
the Institute of Mathematics and its Application (IMA) 
at the University of Minnesota at which much of the discussion
took place. YM thanks Charles S. Peskin for pointing to 
reference \cite{finkelstein1987water} and experimental 
observations on osmosis described therein.

\bibliographystyle{elsarticle-num}
\bibliography{mylib}

\begin{thebibliography}{10}
\expandafter\ifx\csname url\endcsname\relax
  \def\url#1{\texttt{#1}}\fi
\expandafter\ifx\csname urlprefix\endcsname\relax\def\urlprefix{URL }\fi
\expandafter\ifx\csname href\endcsname\relax
  \def\href#1#2{#2} \def\path#1{#1}\fi

\bibitem{pappenheimer1987silver}
J.~Pappenheimer, {A silver spoon}, Annual review of physiology 49~(1) (1987)
  1--16.

\bibitem{boron2008medical}
W.~Boron, E.~Boulpaep, {Medical physiology}, 2nd Edition, W.B. Saunders, 2008.

\bibitem{davson1970textbook}
H.~Davson, A Textbook of General Physiology, 4th Edition, Churchill, 1970.

\bibitem{somjen2004ions}
G.~Somjen, {Ions in the Brain}, Oxford University Press, 2004.

\bibitem{kahle2009molecular}
K.~Kahle, J.~Simard, K.~Staley, B.~Nahed, P.~Jones, D.~Sun, {Molecular
  mechanisms of ischemic cerebral edema: role of electroneutral ion transport},
  Physiology 24~(4) (2009) 257.

\bibitem{hill2008fluid}
A.~Hill, {Fluid Transport: A Guide for the Perplexed}, Journal of Membrane
  Biology 223~(1) (2008) 1--11.

\bibitem{koeppen2007renal}
B.~Koeppen, B.~Stanton, {Renal physiology}, The Mosby physiology monograph
  series, Mosby Elsevier, 2007.

\bibitem{layton2009mammalian}
A.~Layton, H.~Layton, W.~Dantzler, T.~Pannabecker, {The mammalian urine
  concentrating mechanism: hypotheses and uncertainties}, Physiology 24~(4)
  (2009) 250.

\bibitem{mathias2007lens}
R.~Mathias, J.~Kistler, P.~Donaldson, {The lens circulation}, Journal of
  Membrane Biology 216~(1) (2007) 1--16.

\bibitem{fischbarg2005mathematical}
J.~Fischbarg, F.~Diecke, {A mathematical model of electrolyte and fluid
  transport across corneal endothelium}, Journal of Membrane Biology 203~(1)
  (2005) 41--56.

\bibitem{taiz2010plant}
L.~Taiz, E.~Zeiger, {Plant Physiology}, Sinauer Associates, Incorporated, 2010.

\bibitem{Hille}
B.~Hille, Ion Channels of Excitable Membranes, 3rd Edition, Sinauer Associates,
  2001.

\bibitem{aidley_physiology_1998}
D.~Aidley, The Physiology of Excitable Cells, 4th Edition, Cambridge University
  Press, New York, 1998.

\bibitem{KS}
J.~Keener, J.~Sneyd, Mathematical Physiology, Springer-Verlag, New York, 1998.

\bibitem{hoppensteadt2002modeling}
F.~Hoppensteadt, C.~Peskin, {Modeling and simulation in medicine and the life
  sciences}, Springer Verlag, 2002.

\bibitem{weinstein1994mathematical}
A.~Weinstein, {Mathematical models of tubular transport}, Annual review of
  physiology 56~(1) (1994) 691--709.

\bibitem{yi2003mathematical}
C.~Yi, A.~Fogelson, J.~Keener, C.~Peskin, {A mathematical study of volume
  shifts and ionic concentration changes during ischemia and hypoxia}, Journal
  of Theoretical Biology 220~(1) (2003) 83--106.

\bibitem{shapiro2001osmotic}
B.~Shapiro, {Osmotic forces and gap junctions in spreading depression: a
  computational model}, Journal of Computational Neuroscience 10~(1) (2001)
  99--120.

\bibitem{lee2008immersed}
P.~Lee, {The immersed boundary method with advection-electrodiffusion}, Ph.D.
  thesis, NEW YORK UNIVERSITY (2008).

\bibitem{mathias1985steady}
R.~Mathias, {Steady-state voltages, ion fluxes, and volume regulation in
  syncytial tissues}, Biophysical journal 48~(3) (1985) 435--448.

\bibitem{katzir1965nonequilibrium}
A.~Katzir-Katchalsky, P.~Curran, {Nonequilibrium thermodynamics in biophysics},
  Harvard University Press, 1965.

\bibitem{kedem1958thermodynamic}
O.~Kedem, A.~Katchalsky, {Thermodynamic analysis of the permeability of
  biological membranes to non-electrolytes}, Biochimica et Biophysica Acta 27
  (1958) 229--246.

\bibitem{kjelstrup2008non}
S.~Kjelstrup, D.~Bedeaux, {Non-equilibrium thermodynamics of heterogeneous
  systems}, World Scientific Pub Co Inc, 2008.

\bibitem{atzberger2009microfluidic}
P.~Atzberger, S.~Isaacson, C.~Peskin, {A microfluidic pumping mechanism driven
  by non-equilibrium osmotic effects}, Physica D: Nonlinear Phenomena 238~(14)
  (2009) 1168--1179.

\bibitem{layton2006modeling}
A.~Layton, {Modeling water transport across elastic boundaries using an
  explicit jump method}, SIAM Journal on Scientific Computing 28 (2006) 2189.

\bibitem{doi1996introduction}
M.~Doi, H.~See, {Introduction to polymer physics}, Oxford University Press,
  USA, 1996.

\bibitem{degroot1962non}
S.~De~Groot, P.~Mazur, {Non-equilibrium thermodynamics}, Vol.~17, Amsterdam
  North Holland, 1962.

\bibitem{hille1982transport}
B.~HILLE, {Transport across cell membranes: carrier mechanisms}, Excitable
  tissues and reflex control of muscle (1982) 47.

\bibitem{tosteson1989membrane}
D.~Tosteson (Ed.), {Membrane transport: people and ideas}, People and Ideas
  Series/American Physiological Society Book, American Physiological Society,
  1989.

\bibitem{tyrrell1971diffusion}
H.~Tyrrell, K.~Harris, {Diffusion in liquids}, in: Diffusion processes:
  proceedings of the Thomas Graham Memorial Symposium, University of
  Strathclyde, Gordon \& Breach Publishing Group, 1971, p.~67.

\bibitem{justice1983conductance}
J.~Justice, {Conductance of electrolyte solutions}, Comprehensive Treatise of
  Electrochemistry 5.

\bibitem{hoheisel1993theoretical}
C.~Hoheisel, J.~Winkelmann, {Theoretical treatment of liquids and liquid
  mixtures}, Elsevier, 1993.

\bibitem{taylor1993multicomponent}
R.~Taylor, R.~Krishna, {Multicomponent mass transfer}, Wiley-Interscience,
  1993.

\bibitem{accascina1959electrolytic}
F.~Accascina, R.~Fuoss, Electrolytic Conductance, Interscience, New York, 1959.

\bibitem{fawcett2004liquids}
W.~Fawcett, {Liquids, solutions, and interfaces: from classical macroscopic
  descriptions to modern microscopic details}, Oxford University Press, USA,
  2004.

\bibitem{lee2008molecular}
L.~Lee, Molecular Thermodynamics of Electrolyte Solutions, World Scientific Pub
  Co Inc, 2008.

\bibitem{kunz2010specific}
W.~Kunz, {Specific Ion Effects}, World Scientific, 2010.

\bibitem{fraenkel2010simplified}
D.~Fraenkel, {Simplified electrostatic model for the thermodynamic excess
  potentials of binary strong electrolyte solutions with size-dissimilar ions},
  Molecular Physics 108~(11) (2010) 1435--1466.

\bibitem{eisenberg2010crowded}
B.~Eisenberg, Crowded charges in ion channels, Advances in Chemical PhysicsIn
  press.

\bibitem{zhang2010molecular}
C.~Zhang, S.~Raugei, B.~Eisenberg, P.~Carloni, {Molecular Dynamics in
  Physiological Solutions: Force Fields, Alkali Metal Ions, and Ionic
  Strength}, Journal of Chemical Theory and Computation (2010) 2--4.

\bibitem{eisenberg2010energy}
B.~Eisenberg, Y.~Hyon, C.~Liu, {Energy variational analysis of ions in water
  and channels: Field theory for primitive models of complex ionic fluids}, The
  Journal of Chemical Physics 133 (2010) 104104.

\bibitem{chen1997permeation}
D.~Chen, L.~Xu, A.~Tripathy, G.~Meissner, B.~Eisenberg, {Permeation through the
  calcium release channel of cardiac muscle}, Biophysical journal 73~(3) (1997)
  1337--1354.

\bibitem{eisenberg1999structure}
R.~Eisenberg, {From structure to function in open ionic channels}, Journal of
  Membrane Biology 171~(1) (1999) 1--24.

\bibitem{gillespie2002physical}
D.~Gillespie, R.~Eisenberg, {Physical descriptions of experimental selectivity
  measurements in ion channels}, European Biophysics Journal 31~(6) (2002)
  454--466.

\bibitem{HH}
A.~Hodgkin, A.~Huxley, A quantitative description of the membrane current and
  its application to conduction and excitation in nerve, Journal of Physiology
  117 (1952) 500--544.

\bibitem{Rubinstein}
I.~Rubinstein, Electro-Diffusion of Ions, SIAM, 1990.

\bibitem{roosbroeck_theory_1950}
W.~V. Roosbroeck, Theory of flow of electrons and holes in germanium and other
  semiconductors, Bell System Tech. J. 29 (1950) 560--607.

\bibitem{Jerome_semiconductor}
J.~Jerome, Analysis of Charge Transport: A Mathematical Study of Semiconductor
  Devices, Springer-Verlag, 1995.

\bibitem{selberherr1984analysis}
S.~Selberherr, {Analysis and simulation of semiconductor devices},
  Springer-Verlag, 1984.

\bibitem{eisenberg1996computing}
R.~Eisenberg, {Computing the field in proteins and channels}, Journal of
  Membrane Biology 150~(1) (1996) 1--25.

\bibitem{bazant2004diffuse}
M.~Bazant, K.~Thornton, A.~Ajdari, {Diffuse-charge dynamics in electrochemical
  systems}, Physical Review E 70~(2) (2004) 21506.

\bibitem{leonetti_biomembrane_1998}
M.~L\'eonetti, On biomembrane electrodiffusive models, European Physical
  Journal B 2 (1998) 325--340.

\bibitem{schaff1997general}
J.~Schaff, C.~Fink, B.~Slepchenko, J.~Carson, L.~Loew, {A general computational
  framework for modeling cellular structure and function}, Biophysical journal
  73~(3) (1997) 1135--1146.

\bibitem{choi1999electrodiffusion}
Y.~Choi, D.~Resasco, J.~Schaff, B.~Slepchenko, {Electrodiffusion of ions inside
  living cells}, IMA Journal of Applied Mathematics 62~(3) (1999) 207.

\bibitem{Sinica}
Y.~Mori, J.~Jerome, C.~Peskin, {A three-dimensional model of cellular
  electrical activity}, Bulletin-Institute of Mathematics Academia Sinica 2~(2)
  (2007) 367--390.

\bibitem{mori_ephaptic_2008}
Y.~Mori, G.~I. Fishman, C.~S. Peskin, Ephaptic conduction in a cardiac strand
  model with {3D} electrodiffusion, Proceedings of the National Academy of
  Sciences 105~(17) (2008) 6463.

\bibitem{morithesis}
Y.~Mori, A three-dimensional model of cellular electrical activity, Ph.D.
  thesis, New York University (2006).

\bibitem{mori2009three}
Y.~Mori, {From three-dimensional electrophysiology to the cable model: an
  asymptotic study}, Arxiv preprint arXiv:0901.3914.

\bibitem{CAMCoS}
Y.~Mori, C.~Peskin, A numerical method for cellular electrophysiology based on
  the electrodiffusion equations with internal boundary conditions at internal
  membranes, Communications in Applied Mathematics and Computational Science 4
  (2009) 85--134.

\bibitem{tosteson1960regulation}
D.~Tosteson, J.~Hoffman, {Regulation of cell volume by active cation transport
  in high and low potassium sheep red cells}, The Journal of general physiology
  44~(1) (1960) 169.

\bibitem{tosteson1964regulation}
D.~Tosteson, {Regulation of cell volume by sodium and potassium transport}, The
  cellular functions of membrane transport (1964) 3--22.

\bibitem{jakobsson1980interactions}
E.~Jakobsson, {Interactions of cell volume, membrane potential, and membrane
  transport parameters}, American Journal of Physiology- Cell Physiology
  238~(5) (1980) C196.

\bibitem{strieter_volume-activated_1990}
J.~Strieter, J.~L. Stephenson, L.~G. Palmer, A.~M. Weinstein, Volume-activated
  chloride permeability can mediate cell volume regulation in a mathematical
  model of a tight epithelium., The Journal of General Physiology 96~(2) (1990)
  319.

\bibitem{lew_behaviour_1979}
V.~L. Lew, H.~G. Ferreira, T.~Moura, The behaviour of transporting epithelial
  cells. i. computer analysis of a basic model, Proceedings of the Royal
  Society of London. Series B, Biological Sciences 206~(1162) (1979) 53--83.

\bibitem{koch_biophysics_1999}
C.~Koch, Biophysics of Computation, Oxford University Press, New York, 1999.

\bibitem{hernandez_modeling_1998}
J.~A. Hernandez, E.~Cristina, Modeling cell volume regulation in nonexcitable
  cells: the roles of the na+ pump and of cotransport systems, American Journal
  of Physiology- Cell Physiology 275~(4) (1998) C1067.

\bibitem{finkelstein1987water}
A.~Finkelstein, Water movement through lipid bilayers, pores, and plasma
  membranes: theory and reality, Distinguished lecture series of the Society of
  General Physiologists, Wiley, 1987.

\bibitem{onsager1931reciprocal}
L.~Onsager, {Reciprocal Relations in Irreversible Processes. II.}, Physical
  Review 38~(12) (1931) 2265--2279.

\bibitem{doi1988theory}
M.~Doi, S.~Edwards, {The theory of polymer dynamics}, International series of
  monographs on physics, Clarendon Press, 1988.

\bibitem{doi2009gel}
M.~Doi, {Gel dynamics}, J. Phys. Soc. Jpn 78 (2009) 052001.

\bibitem{hyon2010energetic}
Y.~Hyon, D.~Kwak, C.~Liu, {Energetic variational approach in complex fluids:
  Maximum dissipation principle}, DCDS-A 26~(4) (2010) 1291--1304.

\bibitem{goldstein1980classical}
H.~Goldstein, {Classical mechanics}, Addison-Wesley series in physics,
  Addison-Wesley Pub. Co., 1980.

\bibitem{biler2000long}
P.~Biler, J.~Dolbeault, {Long time behavior of solutions to Nernst-Planck and
  Debye-H{\\"u}ckel drift-diffusion systems}, Annales Henri Poincar{\'e} 1~(3)
  (2000) 461--472.

\bibitem{ryham2009existence}
R.~Ryham, {Existence, Uniqueness, Regularity and Long-term Behavior for
  Dissipative Systems Modeling Electrohydrodynamics}, Arxiv preprint
  arXiv:0910.4973.

\bibitem{leonetti_pattern_2004}
M.~L\'eonetti, E.~{Dubois-Violette}, F.~Hombl\'e, Pattern formation of
  stationaly transcellular ionic currents in fucus, Proc. Natl. Acad. Sci.
  {USA} 101~(28) (2004) 10243--10248.

\bibitem{brera2010conservative}
M.~Brera, J.~Jerome, Y.~Mori, R.~Sacco, {A Conservative and monotone
  mixed-hybridized finite element approximation of transport problems in
  heterogeneous domains}, Computer Methods in Applied Mechanics and Engineering
  199 (2010) 2709--2720.

\bibitem{ibmethod}
C.~Peskin, The immersed boundary method, Acta Numerica 11 (2002) 479--517.

\bibitem{li_immersed_2006}
Z.~Li, K.~Ito, The Immesed Interface Method: Numerical Solutions of {PDEs}
  Involving Interfaces and Irrgular Domains, {SIAM}, 2006.

\bibitem{eisenberg_three-dimensional_1970}
R.~S. Eisenberg, E.~A. Johnson, Three-dimensional electrical field problems in
  physiology, Prog. Biophys. Mol. Biol 20~(1).

\bibitem{eisenberg1979electrical}
R.~Eisenberg, V.~Barcilon, R.~Mathias, {Electrical properties of spherical
  syncytia}, Biophysical Journal 25~(1) (1979) 151--180.

\bibitem{mathias1979electrical}
R.~Mathias, J.~Rae, R.~Eisenberg, {Electrical properties of structural
  components of the crystalline lens}, Biophysical Journal 25~(1) (1979)
  181--201.

\bibitem{neu1993hst}
J.~Neu, W.~Krassowska, {Homogenization of syncytial tissues}, Critical reviews
  in biomedical engineering 21~(2) (1993) 137--199.

\end{thebibliography}

\end{document}